\documentclass[10pt,a4paper]{article}

\usepackage[utf8]{inputenc} 
\usepackage[T1]{fontenc}    %
\usepackage[english]{babel} 
\usepackage[final]{microtype} 

\usepackage{amssymb,amsmath,mathtools} 
\usepackage{amsthm} 

\usepackage{xcolor}

\usepackage[hmargin=0.12\paperwidth,vmargin=0.19\paperwidth,bindingoffset=0cm]{geometry} 

\numberwithin{equation}{section} 

\usepackage{titlesec}

\titleformat*{\section}{\large\bfseries}
\titleformat*{\subsection}{\bfseries}

\theoremstyle{plain}
  \newtheorem{thm}{Theorem}
  \newtheorem{prop}[thm]{Proposition}
  \newtheorem{lemma}[thm]{Lemma}
  
\theoremstyle{definition}
  
  \newtheorem{remark}[thm]{Remark}

\newcommand{\N}{\mathbb{N}} \newcommand{\Z}{\mathbb{Z}} 
\newcommand{\R}{\mathbb{R}} \newcommand{\C}{\mathbb{C}}
 \newcommand{\E}{\mathbb{E}}

\DeclareMathOperator{\tr}{Tr} 
\DeclareMathOperator{\sign}{sgn} 
\DeclareMathOperator{\erfc}{erfc} 

\DeclarePairedDelimiter{\abs}{\lvert}{\rvert} 

\newcommand{\MeijerG}[8][\Big]{G^{{ #2 },{ #3 }}_{{ #4 },{ #5 }} #1( \begin{matrix} #6 \\ #7 \end{matrix}\, #1\vert\, #8 #1)}
\newcommand{\hypergeometric}[6][\Big]{\,{}_{#2} F_{#3} #1( \begin{matrix} #4 \\ #5 \end{matrix}\, #1\vert\, #6 #1)}

\usepackage{hyperref}   
\hypersetup{
pdfborder={0 0 0},      
}

\title{\Large\bfseries Real eigenvalue statistics \\
for products of asymmetric real Gaussian matrices}
\author{Peter J. Forrester and Jesper R.  Ipsen}
\date{}

\begin{document}

\maketitle

\begin{center}
\small\itshape
Department of Mathematics and Statistics,\\ ARC Centre of Excellence for Mathematical
 and Statistical Frontiers,\\ The University of Melbourne,
Victoria 3010, Australia
\end{center}

\bigskip

\begin{abstract}
\noindent 
Random matrices formed from i.i.d.~standard real
Gaussian entries have the feature that the expected number of real eigenvalues 
is non-zero. This property persists for products of such matrices, independently chosen,
and moreover it is known that as the number of matrices in the product tends to infinity,
the probability that all eigenvalues are real tends to unity. We quantify the distribution
of the number of real eigenvalues for products of finite size real Gaussian matrices by giving an
explicit Pfaffian formula for the probability that there are exactly $k$ real eigenvalues as a determinant
 with entries involving particular Meijer $G$-functions. We also compute
the explicit form of the Pfaffian correlation kernel for the correlation between real eigenvalues,
and the correlation between complex eigenvalues. The simplest example of these --- the eigenvalue
density of the real eigenvalues --- gives by integration the expected number of real eigenvalues.
Our ability to perform these calculations relies on the construction of certain skew-orthogonal polynomials
in the complex plane, the computation of which is carried out using their relationship to particular random
matrix averages.
\end{abstract}

\vspace{3em}

\section{Introduction}


A basic question in random matrix theory is to ask for the probability distribution of the number of real
eigenvalues for an ensemble of $N \times N$ random matrices with real entries. With the ensemble made
up of standard Gaussian random matrices, i.e.~in the circumstance that each element is independently chosen
as a real standard Gaussian, Edelman  \cite{Ed97} was the first person to obtain results on this problem.
The approach taken centered on knowledge of the explicit functional form of the
probability density function (PDF) for the event that there are $k$ real eigenvalues denoted $\{\lambda_l\}_{l=1}^k$, and $N-k$ complex
eigenvalues denoted $\{x_j \pm i y_j \}_{j=1}^{(N-k)/2}$ with $(x_j,y_j) \in \R\times\R_+$ (the fact that the complex
eigenvalues occur in complex conjugate pairs implies $k$ must have the same parity as $N$).
Thus it was shown that this is equal to
\begin{equation}\label{3.1}
\frac{1}{ k! ((N-k)/2)! } \frac1{Z_{N}}
\abs*{ \Delta\Big(\{\lambda_l\}_{l=1}^k \cup \{ x_j \pm i y_j \}_{j=1}^{(N-k)/2}\Big) } 
\prod_{j=1}^ke^{-  \lambda_j^2/2} 
\prod_{j=1}^{(N-k)/2} 2e^{y_j^2 - x_j^2}\erfc(\sqrt{2} y_j),
\end{equation}
where $\Delta(\{z_p\}_{p=1}^m) := \prod_{j < l}^m (z_l - z_j)$ denotes the Vandermonde determinant and
\begin{equation}\label{normalisation}
Z_{N}=2^{N(N+1)/4} \prod_{l=1}^N \Gamma(l/2).
\end{equation}
(see also \cite{LS91}).
Integrating (\ref{3.1}) over $\{\lambda_l\} \cup \{x_j + i y_j \}$ gives the probability $p_{N,k}$ that there are
exactly $k$ real eigenvalues.
The simplest case to compute is when $k=N$ and thus all eigenvalues are real,
for which the probability was found to equal  $2^{-N(N-1)/4}$.

Questions relating to the probability that all eigenvalues are real for random matrices with real entries occur in applications. 
Consider first the tensor structure $\mathcal A = (a_{ijk}) \in
\mathbb R^{p \times p \times 2}$, represented as the column vector
${\rm vec } \, \mathcal A \in \mathbb R^{4 p^2}$. As reviewed in \cite{KB09}, it is of interest to find matrices
$U = [ \vec{u}_1 \cdots \vec{u}_R] \in \mathbb R^{p \times R}$,
$V = [ \vec{v}_1 \cdots \vec{v}_R] \in \mathbb R^{p \times R}$,
$W = [ \vec{w}_1 \cdots \vec{w}_R] \in \mathbb R^{2 \times R}$ such that
\[
{\rm vec} \, \mathcal A = \sum_{r=1}^R \vec{w}_r \otimes \vec{v}_r \otimes \vec{u}_r
\]
for $R$ --- referred to as the rank --- as small as possible. 
It turns out that with both $(a_{ij1}) =: X_1 \in \mathbb R^{p \times p}$ and
$(a_{ij2}) =: X_2 \in \mathbb R^{p \times p}$ random matrices, entries chosen from a continuous
distribution, one has that $R = p$ if all the eigenvalues of $X_1^{-1} X_2$ are real,
and $R = p + 1$ otherwise \cite{tB91}. In the Gaussian case these probabilities have been computed in
\cite{FM11} and \cite{BF12}
as equal to $(\Gamma((p+1)/2))^p/G(p+1)$, 
where $G(x)$ denotes the Barnes $G$-function,
and the corresponding large $R$ asymptotic form has been computed in \cite{BF12}.

A second example comes from quantum entanglement. Lakshminarayan \cite{La13} considered the problem of
quantifying when two-qubits $|\phi_1\rangle$ and $|\phi_2\rangle$ are an optimal pair, in the case that the states
are chosen from a uniform distribution on the unit 3--sphere. The condition of being an optimal pair is known
\cite{SRL11} as  particular inequalities for certain weighted inner products between the qubits. It was shown
in \cite{La13} that these can be interpreted as the condition for the probability that the random matrix product
$X_1 X_2$, with each $X_i$ a $2 \times 2$ real Gaussian matrix, having all eigenvalues real, which
was furthermore shown to be equal to $\pi/4$.

An intriguing effect was observed in the study \cite{La13}, which seems to have escaped early notice. Thus,
noting from the result of Edelman cited above that for a single real  Gaussian $2 \times 2$ random matrix
the probability of all eigenvalues being real is equal to $2^{-1/2}$, while for a product of two independent Gaussian
$2 \times 2$ random matrices it is $\pi/4$, the fact that $2^{-1/2} < \pi/4$ led Lakshminarayan
to investigate if the probability
of all eigenvalues being real was an increasing function of the number of matrices in the product. Numerical simulation
indicated that this is indeed the case, and further the probability that all eigenvalues are real tends to unity
as the number of random matrices in the product tends to infinity. Evidence that this is also true for products of $d \times d$
real Gaussian matrices was given in \cite{La13}, while the follow up work \cite{HJL15} provided similar evidence for random matrices with independent non-Gaussian entries. A proof in the instance of the latter circumstance that the entries are all independent
and identically distributed with a PDF containing an atom has recently been given in \cite{Re16}.

The appearance of the work  \cite{La13} coincided with the appearance of works containing
other surprising advances relating to the
eigenvalues of products of random matrices. Consider the random matrix product 
\begin{equation}\label{YX}
P_m=X_1\cdots X_m
\end{equation}
 where
each $X_i$ is an $N \times N$ standard Gaussian matrix. In the case of complex entries, Akemann and Burda \cite{AB12} showed that
the eigenvalues form a determinantal point process in the complex plane. This means that the $k$-point correlation
function for the eigenvalues $\rho_{(k)}(z_1,\dots, z_k)$ is fully determined by a single function $K(w,z)$,
referred to as the correlation kernel, according to
\begin{equation}
\rho_{(k)}(z_1,\dots, z_k) = \det [ K(z_j, z_l) ]_{j,l=1,\dots, k}.
\end{equation}
In the case of real entries, Forrester \cite{Fo13}  found a closed form expression for the probability that all
eigenvalues are real.

To specify this latter result requires introducing
the Meijer $G$-function
\begin{equation}\label{GG}
\MeijerG{m}{n}{p}{q}{a_1,\ldots,a_p}{b_1,\ldots,b_q}{z}
=
\frac{1 }{ 2 \pi i}
\int_\gamma \frac{\prod_{j=1}^m \Gamma ( b_j - s) \prod_{j=1}^n \Gamma(1 - a_j + s)}
{\prod_{j=m+1}^q \Gamma (1-  b_j + s) \prod_{j=n+1}^p \Gamma(a_j - s) } z^s \, ds,
\end{equation}
where  $\gamma$ is an appropriate contour relating to the validity of the inverse Mellin transform formula.
With $p_{N,k}^{P_m}$ denoting the probability that the random matrix product $P_m$  (\ref{YX}) has exactly
$k$ real eigenvalues it was shown in  \cite{Fo13} that, for each $X_i$ a real Ginibre matrix, we have
\begin{equation}\label{11}
p_{N,N}^{P_m} = \Big (\prod_{j=1}^N\frac{1}{\Gamma(j/2)}\Big)^m\times
\begin{cases}
\det\Big[ [a_{j,k}]_{k=1,\ldots,N/2}^{j=1,\ldots,N/2}\Big], & N\ \text{even}\\[2mm]
\det\Big[ [a_{j,k}]_{k=1,\ldots,(N-1)/2}^{j=1,\ldots,(N+1)/2}\quad 
[\tilde a_j]_{j=1,\ldots,(N+1)/2}\Big], & N\ \text{odd}
\end{cases}
\end{equation}
with
\begin{equation}\label{aa}
a_{j,k} 
= \MeijerG[\bigg]{m+1}{m}{m+1}{m+1}{\tfrac32-j,\dots, \tfrac32-j,1}{0, k,\dots, k}{1}
\qquad\text{and}\qquad
\tilde a_j=\Gamma(j - {1/2})^m.
\end{equation}
%
A simple identity for the Meijer $G$-function --- evident from the definition (\ref{GG}) --- shows that $a_{j,k}$ is
equal to the Meijer $G$-function occurring in~\cite{Fo14}. 
Moreover, these explicit formulas were used to prove that $p_{N,N}^{P_m}  \to 1$ as $m \to \infty$. Extension to rectangular matrices where given in~\cite{IK14,Ip15,Ip15thesis}, while special arithmetic properties were shown to be present in the case $m = 2$ \cite{Ku15}.

A primary aim of the  present paper is to  extend this result to the
calculation of $p_{N,k}^{P_m}$, for general $0 \le k \le N$ with the same parity as $N$. 
The following theorem will be proved in section~\ref{sec:k-real}.

\begin{thm}\label{T1}
Consider the random matrix product (\ref{YX}), in which each $X_i$ is a real Ginibre matrix. Let
\begin{equation}
b_{j,l}(\zeta):=(\zeta - 1) \Big ( a_{j,l} - 2^{-2} (2 (l-1))^m a_{j,l-1} \Big ) +
2^{-(2j-1/2)m} h_{j-1}\delta_{j,l} 
\end{equation}
with $h_{j}=(2 \sqrt{2 \pi} \Gamma(2j+1))^m$, $a_{j,l}\ (l>0)$ given by~\eqref{aa} and $a_{j,0}=0$. For $N$ even, the probability $p_{N,2k}^{P_m}$ that exactly $2k$ eigenvalues are real is given by
\begin{equation}\label{a1}
p_{N,2k}^{P_m} = \Big ( \prod_{j=1}^N \frac{1 }{ \Gamma (j/2) } \Big )^m [\zeta^k]
\det \Big [ b_{j,l}(\zeta)  \Big ]_{j,l=1,\dots,N/2},
\end{equation}
while for $N$ odd we have
\begin{equation}\label{14a}
p_{N,2k+1}^{P_m} = \Big ( \prod_{j=1}^N \frac{1 }{ \Gamma (j/2) } \Big )^m   [\zeta^k]
\det \Big [ [b_{j,l}(\zeta) ]^{j=1,\dots,(N+1)/2}_{k=1,\dots,(N-1)/2}  
\quad [\tilde a_j ]_{j=1,\dots,(N+1)/2} \Big]
\end{equation}
with $\tilde a_j$ from~\eqref{aa}. In both~\eqref{a1} and~\eqref{14a} $\zeta$ is a generating function parameter for the probabilities and $[\zeta^k] f(\zeta) $ denotes the coefficient of $\zeta^k$ in the power series expansion of $f(\zeta)$, i.e. for $N$ even
\[
1=\sum_{k=0}^{N/2}p_{N,2k}^{P_m}= \Big(\prod_{j=1}^N\frac{1}{\Gamma(j/2)}\Big)^m\det\Big[b_{j,l}(1)\Big]_{j,l=1,\dots,N/2}
\]
and similarly for the $N$ odd case in terms of~\eqref{14a}.

\end{thm}

A formula
closely related to Theorem \ref{T1} in the case $m=1$ was derived by Akemann and Kanzieper \cite{AK05}, and
this working was soon after refined \cite{AK07} to obtain a formula equivalent to (\ref{a1}). Also for this case
Forrester and Nagao \cite{FN08p} gave a result more general than (\ref{a1}), applying to a
real random matrix formed from a general linear combination of Gaussian symmetric and anti-symmetric
matrices.

We now turn our attention to the other primary aim of our work. This relates to the statistical state formed by
the eigenvalues of the product (\ref{YX}). In the complex case, it has been remarked that the statistical state
is a determinantal point process. In the real case, it is known
from the work of Ipsen and Kieburg
 \cite{IK14} that the eigenvalue correlations form instead a Pfaffian point process. Thus, considering for
definiteness the real eigenvalues, one now has
\begin{equation}\label{K0}
\rho_{(k)}^{{\rm real} }(x_1,\dots, x_k) = {\rm Pf} \, [ \mathbf K^{\rm rr}(x_j, x_l) ]_{j,l=1,\dots, k}, \qquad
\mathbf K^{\rm rr} (x,y) = \begin{bmatrix} D(x,y) & S(x,y) \\
- S (y,x) & \tilde{I}(x,y) \end{bmatrix},
\end{equation}
where $D(x,y)$ and $\tilde{I}(x,y)$ are antisymmetric functions of $x$ and $y$.

A concern of the present paper is to compute the explicit form of the correlation kernel in (\ref{K0}) in the
case of the real eigenvalues of (\ref{YX}) for real standard Gaussian matrices, and also for the case of the complex eigenvalues. In this paper, we will see that these correlation kernels possess many similarities with other results for product of random matrices.
%
For example, the kernel for the Pfaffian point process specifying the scaled statistical state about the origin of the real
eigenvalues of  products of real Ginibre matrices is given in terms of Meijer $G$-functions. 
In the simplest case of the one point function $\rho_{(1)}^r(x) $ the resulting functional form is very succinct.
\begin{thm}\label{T2}
 Define
\begin{equation}
w_r(\lambda) =\MeijerG{m}{0}{0}{m}{\underline{\hspace{0.5cm}}}{0, \dots, 0}{\frac{\lambda^2}{2^m}}=
\prod_{j=1}^m \bigg[\int_{-\infty}^\infty d \lambda^{(j)}e^{-(\lambda^{(j)})^2/2}\bigg] \,
 \delta( \lambda - \lambda^{(1)} \cdots \lambda^{(m)}).
\end{equation}
We have
\begin{equation}\label{r-local}
\lim_{N\to\infty}\rho^r_{(1)}(x)= 
\int_{-\infty}^\infty dv\,\abs{x-v} w_r(x)w_r(v)
\MeijerG{1}{0}{0}{m}{-}{0,\ldots,0}{-xv}.
\end{equation}
\end{thm}

For singular values
of products of complex Ginibre matrices, it is similarly the case that the kernel for the scaled determinantal
point process in the neighbourhood of the origin can be expressed in terms of Meijer $G$-functions \cite{KZ14}; see also the recent review
\cite{AI15}.
Moreover, for fixed $N$, knowledge of the real-to-real eigenvalue correlations gives information about the moments of the distribution function for the
probability that there are $k$ real eigenvalues. In particular, integration of the spectral density (one-point function) gives the
expected number of real eigenvalues.

The rest of this paper is organised as follows. In section~\ref{sec:jpdf} we find the joint eigenvalue PDF for a Gaussian product matrix with a given number of real eigenvalues. In section~\ref{sec:skew} we introduce the generalised partition function and find the skew-orthogonal polynomials; we combine these results with the joint eigenvalues PDF to prove Theorem~\ref{T1}. Section~\ref{sec:corr} focuses on the real-to-real and the complex-to-complex eigenvalue correlations. In particularly, we study local and global scaling limits for the spectral densities and use the real spectral density to compute the expected number of real eigenvalues. The final section briefly sketches how all these results may be extended to products of rectangular matrices.

\section{Joint probability density function}
\label{sec:jpdf}

Our first task is to find the explicit functional form for the eigenvalue PDF of the random matrix product
(\ref{YX}) in the case that each $X_i$ is an independent $N \times N$ standard real Gaussian matrix. 
With this specification the joint probability measure for $\{P_m,X_1,\dots,X_m\}$ is equal to
\begin{equation}\label{Z1}
\delta(P_m - X_1 \cdots X_m)  
\bigg(\prod_{l=1}^m \Big ( \frac{1}{2 \pi} \Big )^{N^2/2} e^{-\frac{1}{2} {\rm Tr} \, X_l X_l^T} \, (d X_l)\bigg)
(d P_m).
\end{equation}
Actually this task, extended to the general bi-orthogonal invariant ensembles, has already been addressed by Ipsen and Kieburg \cite{IK14}. However the workings therein
are not sufficient for all our purposes. In particular proportionality constants are ignored, meaning that it is not
possible to proceed to derive the formulas of Theorem \ref{T1} for the probabilities 
$p_{N,k}^{P_m}$. These normalisation constants were included in the thesis~\cite{Ip15thesis} but the PDF were given in terms of $2\times 2$ matrices, which is impractical for our purpose. Furthermore, the case that the working of \cite{IK14,Ip15thesis} --- which is a generalisation of the strategies used in \cite{So07} and \cite{APS10} in the cases $m=1$ and $m=2$ respectively --- treats the real and complex eigenvalues on an equal footing, whereas we prefer to proceed in the way used in \cite{Ed97} for $m=1$ which distinguishes the real and complex eigenvalues from the outset. Below we give a more practical formulation of the joint eigenvalue PDF. 

\begin{thm}\label{T3}
Let
 \begin{equation}\label{C1a}
w_r(\lambda) =\MeijerG{m}{0}{0}{m}{\underline{\hspace{0.5cm}}}{0, \dots, 0}{\frac{\lambda^2}{2^m}}=
\prod_{j=1}^m \bigg[\int_{-\infty}^\infty d \lambda^{(j)}e^{-(\lambda^{(j)})^2/2}\bigg] \,
\delta( \lambda - \lambda^{(1)} \cdots \lambda^{(m)}),
\end{equation}
referred to as the real (or one-point) weight function and let
 \begin{equation}\label{L3a}
 w_c(x,y) = 2\pi\, \int_{-\infty}^\infty d\delta\, \frac{ |\delta| }{ \sqrt{\delta^2 + 4 y^2}}\,
  W\Big (\begin{bmatrix} \mu_+ & 0 \\ 0 & \mu_- \end{bmatrix} \Big ),\qquad
  \mu_{\pm} = \frac12 \Big ( \pm | \delta | + [ \delta^2 + 4(x^2 + y^2) ]^{1/2} \Big )
\end{equation}
with 
\begin{equation}\label{L3}
 W(G)  =
 \prod_{l=1}^m\bigg[\int_{\R^{2\times 2}}(dG^{(l)})\frac{e^{-\frac12 \tr G^{(l)}G^{(l)T}}}{\sqrt{2\pi^3}}\bigg]
 \delta(G - G^{(1)} \cdots G^{(m)}),
\end{equation}
referred to as the complex (or two-point) weight function.

Consider the product~\eqref{YX}. Given that there are $k$ real eigenvalues ($k$ of the same parity as the matrix dimension $N$), the joint 
eigenvalue PDF is
\begin{equation}\label{3.1a}
\frac{1}{ k! ((N-k)/2)! } \Big ( \frac1{Z_{N}} \Big )^m
\abs[\Big]{ \Delta\Big(\{\lambda_l\}_{l=1}^{k} \cup \{ x_j \pm i y_j \}_{j=1}^{(N-k)/2}\Big) }
\prod_{j=1}^k w_r(\lambda_j)\prod_{j=k+1}^{(N+k)/2} w_c(x_j,y_j)
\end{equation}
with $Z_{N}$ given by~\eqref{normalisation} and $w_r,w_c$ as above.
\end{thm}

\begin{proof}
The starting point is to use a generalised real Schur decomposition to triangularise the matrices $\{X_l\}_l$ which appear in the product~\eqref{YX}. Assuming that the product matrix~\eqref{YX} has $k$ real eigenvalues, the decompositions states that for invertible matrices (Gaussian matrices are invertible almost surely) we may write~\cite[Prop. A.26]{Ip15thesis}
\begin{equation}\label{gen-schur}
X_l=Q_l(D_l+T_l)Q_{l+1}^{-1},\qquad l=1,\ldots,m
\end{equation}
with $Q_{m+1}:=Q_1$. Here each $Q_l$ is a real orthogonal matrix in $O^*(N)/O^*(2)^{(N-k)/2}$ with $O^*(N)$ defined to be the set of matrices in $O(N)$ with the first entry in each column positive. Each $D_l$ is a (block) diagonal matrix with
the first $k$ diagonal entries scalars $\{\lambda_1^{(l)},\dots,\lambda_k^{(l)}\}$
and the next $(N-k)/2$ block entries $2 \times 2$ matrices $\{ G_s^{(l)} \}_{s=k+1}^{(N+k)/2}$,
while each $T_l$ is a strictly upper triangular matrix consisting of $N(N-1)/2 - (N-k)/2$ independent Gaussian random variables. 

The generalised Schur decomposition may be verified by applying an ordinary Schur decomposition on the product matrix~\eqref{YX} itself and then using (partial) QR decompositions on $\{Q_l X_l\}_{l=1,\ldots,m-1}$, recursively (see~\cite[Appendix A]{Ip15thesis} for details). We stress that while it is possible to choose $m-1$ of the matrices $D_l$ in~\eqref{gen-schur} to be strictly diagonal rather than block diagonal (due to the $m-1$ QR decompositions), we do not do so as it would complicate the derivation of the Jacobian. 

For the following, it will be convenient to introduce the product $D:=D_1\cdots D_m$ which again is a block diagonal matrix. The first $k$ diagonal entries are scalars, $\{\lambda_t:=\lambda_t^{(1)}\cdots\lambda_t^{(l)}\}_{t=1}^k$, while the latter $(N-k)/2$ entries are $2 \times 2$ matrices, $\{G_s:=G_s^{(1)}\cdots G^{(l)}_s \}_{s=k+1}^{(N+k)/2}$. With this notation, the Jacobian for the above given change of variables reads~\cite[Prop. A.26]{Ip15thesis}
\begin{equation}
\prod_{l=1}^m (d X_l) =
 \prod_{j < p} | \lambda(D_{pp}) - \lambda(D_{jj})|  \prod_{l=1}^m(d T_l) ( Q_l^T dQ_l) 
 \prod_{l=1}^m \Big (\prod_{j=1}^k d \lambda_j^{(l)} \prod_{s=k+1}^{(N+k)/2} dG_s^{(l)} \Big ),
\end{equation}
where $\lambda(D_{pp})$ refers to the eigenvalue(s) of the $(pp)$-th entry of the block diagonal matrix $D$, i.e. $\lambda(D_{pp})=\lambda_p$ for $p=1,\dots,k$ and $\lambda(D_{pp})$ denotes the two (complex) eigenvalues of the $2 \times 2$ block $G_p$ for $p > k$. Thus, using the notation $\{x_j\pm iy_j\}_j$ with $(x_j,y_j)\in\R\times\R_+$ for the complex eigenvalues, we have
\[
 \prod_{j < p} | \lambda(D_{pp}) - \lambda(D_{jj})|=
 \begin{cases}
 \abs{\lambda_j-\lambda_p} & \text{if}\ j<p\leq k,\\
 (\lambda_j-x_p)^2+y_p^2 & \text{if}\ j\leq k<p,\\
  ((x_j-x_p)^2+(y_j-y_p)^2)(x_j-x_p)^2+(y_j+y_p)^2)& \text{if}\ k\leq j<p.\\
 \end{cases}
\]
This notation is the same as used by Edelman~\cite[Eq. (6)]{Ed97}. More compactly, we may write
\[
 \prod_{j < p} | \lambda(D_{pp}) - \lambda(D_{jj})|=
 \abs*{ \Delta\Big(\{\lambda_l\}_{l=1}^k\cup\{x_j\pm iy_j\}_{j=1}^{(N-k)/2}\Big)}\prod_{j=1}^{(N-k)/2}\frac1{2y_j},
\]
where the Vandermonde determinant is defined as in~\eqref{3.1}. 
 
For the weight function in (\ref{Z1}) we have
 \[
 \prod_{l=1}^m e^{-\frac12 {\rm Tr} \, X_l X_l^T}  =
 \prod_{l=1}^m
 e^{- \frac{1}{2} \sum_{s=1}^k (\lambda_s^{(l)})^2 - \frac{1}{2}
 \sum_{s=k+1}^{(N+k)/2} {\rm Tr} \, G_s^{(l)} ( G_s^{(l)})^T}
 e^{- \frac12 \sum_{i < j} (t_{ij}^{(l)})^2},
 \]
where we can integrate out the dependence on $\{T_l\}$ and $\{Q_l\}$ according to
 \[
 \int (d T_l) \, e^{- \frac12 \sum_{i < j} (t_{ij}^{(l)})^2} =
 (2 \pi)^{(N(N-1)/2 - (N-k)/2)/2}\quad\text{and}\quad
 \int (Q_l^T d Q_l) =
 \frac{\pi^{N(N+1)/4-(N-k)/2}}{\prod_{j=1}^N \Gamma(j/2)}.
 \]
The latter is equal to ${{\rm vol} \, O^*(N)}/{({\rm vol} \, O^*(2) )^{(N-k)/2}}$.
 
Using all the above results, it follows that, for a given $k$,
the joint probability measure for the eigenvalues is equal to
\begin{multline}\label{L1}
 \prod_{j < p} | \lambda(D_{pp}) - \lambda(D_{jj}) |\ 
 \prod_{j=1}^k \delta(\lambda_j - \lambda_j^{(1)} \cdots \lambda_j^{(m)} )\, d \lambda_j
 \prod_{s=k+1}^{(N+k)/2} \delta(G_s - G_s^{(1)} \cdots G_s^{(m)})\, (d G_s)\\
 \times
 \prod_{l=1}^m \bigg[
 \frac{1 }{Z_{N}}
 \prod_{j=1}^k \Big( e^{- \frac12  (\lambda_j^{(l)} )^2 }d \lambda_j^{(l)} \Big) 
 \prod_{s=k+1}^{(N+k)/2} \Big( \frac{e^{- \frac12 \tr G_s^{(l)} G_s^{(l) T}}}{\sqrt{2\pi^3}}(dG_s^{(l)}) \Big)
 \bigg ].
\end{multline}
We have, at this point, not yet explicitly introduced the constraint that the eigenvalues of each $G_s$ are not real and thus are consequently a complex conjugate pair. For this reason, we have a similarity with~\cite[Prop.~4.26]{Ip15thesis}.  
 
In order to explicitly impose our constraint that the product matrix has exactly $k$ real eigenvalues, we suppose an orthogonal similarity transformation has been used to bring each matrix $G_i$ into the form
\begin{equation}\label{edel-decomp}
\begin{bmatrix} x & b \\ -c & x \end{bmatrix}
\end{equation}
with $b,c > 0$. The eigenvalues are then $x \pm i y$ with $y^2 = bc$, and we know too (see e.g.~\cite[Proof of Prop.~15.10.1 and Prop.~15.10.2]{Fo10}) that changing variables from the elements
of $G_i$ to $\{x,y,\delta,\theta\}$, where $\theta$ parametrises the orthogonal similarity transformation and
$\delta = b - c$ introduces the Jacobian
\[
 \frac{4\, y\,\abs\delta}{\sqrt{ \delta^2 + 4 y^2}}.
\]
The fact that the integrand in (\ref{L3}) is invariant under real orthogonal transformations
allow us to simplify further. Firstly, we may integrate out $\theta$, which contributes with an extra factor of $\pi$. Secondly, we may replace the matrix $G_i$ by the diagonal matrix of its singular values, $\mu_+$ and $\mu_-$ say.
In terms of the variables $x,y,\delta$ it is straightforward to compute that the singular values are
given by (\ref{L3a}). Combining these results completes the proof.
\end{proof}

Due to the relatively involved expression for the two-point weight~\eqref{L3a}, it might be beneficial to briefly expand on the simplest cases, $m=1$ and $m=2$, where explicit expressions are known.

For $m=1$, the joint PDF (\ref{3.1a}) must, of course, reduce to the classical result (\ref{3.1}).
Inspection of~\eqref{C1a} and~\eqref{L3} shows that the integration therein are immediate for $m=1$ due to the delta functions.
In the real case we then read off that $w_r(\lambda)=e^{-\lambda^2/2}$. 
In the complex case, substituting in (\ref{L3a}) gives
\begin{equation}\label{wc1}
w_c(x,y) =  \sqrt{\frac2\pi}
e^{- (x^2 + y^2)}
\int_{-\infty}^\infty d\delta\, \abs\delta\frac{ e^{- \delta^2/2}}{ \sqrt{\delta^2 + 4 y^2}}  
=  {2} e^{-x^2 + y^2} {\rm erfc} \, (\sqrt{2} y),
\end{equation}
where the second equality first appeared in \cite{Ed97}, albeit out by a factor of 2 as remarked in \cite{Ma11}. 
Substituting these evaluations in~\eqref{3.1a} indeed reproduces~\eqref{3.1}.

Returning now to the case $m=2$, the Meijer $G$-function (\ref{C1a}) is a modified Bessel function,
\begin{equation}\label{m2g}
w_r(\lambda) = 2 K_0(|\lambda|).
\end{equation} 
To simplify (\ref{L3a}) requires simplifying (\ref{L3}).
For this purpose, and without yet restricting $m$, we introduce $2 \times 2$ real matrices $\{M^{(l)}=G^{(l)}\cdots G^{(1)}\}_{l=1,\ldots,m}$ and set $M^{(0)}$ equal to the $2\times 2$ identity matrix. We note that
\[
(d G^{(1)}) \cdots (d G^{(m)})  =
\abs{\det M^{(1)}}^{-1} \cdots \abs{\det M^{(m)}}^{-1}  (dM^{(1)}) \cdots (dM^{(m-1)}),
\]
which allows the integration over $M^{(m)}$ to be carried out in~\eqref{L3} using the delta function, showing that
\begin{equation}\label{L3b}
W(G)  =  \prod_{l=1}^{m-1}\bigg[\int_{\R^{2\times 2}}\frac{(dM^{(l)})}{\abs{\det M^{(l)}}}
 \frac{e^{-\frac12 \tr ((M^{(l-1)}M^{(l-1)T})^{-1}M^{(l)}M^{(l)T})}}{\sqrt{2\pi^3}}\bigg]
 \frac{e^{-\frac12 \tr ((M^{(m-1)}M^{(m-1)T})^{-1}GG^T)}}{\sqrt{2\pi^3}}.
 \end{equation} 
This is the two-by-two matrix version of~\cite[Eq. (2.20)]{Ip15thesis}. 
A further change variables $A^{(l)}=M^{(l)} M^{(l)T}$ for each $l=1,\dots,m-1$ 
shows
\begin{equation}\label{L3c}
W(G)  = 
 \prod_{l=1}^{m-1}\bigg[\int_{A>0}\frac{(dA^{(l)})}{(\det A^{(l)})^{3/2}}
 \frac{e^{-\frac12 \tr ((A^{(l-1)})^{-1}A^{(l)})}}{\sqrt{2\pi}}\bigg]
 \frac{e^{-\frac12 \tr ((A^{(m-1)})^{-1}GG^T)}}{\sqrt{2\pi^3}},
\end{equation} 
where the integration is over positive-definite real symmetric matrices $A^{(l)}$, $l=1,\dots,m-1$.
In the case $m=2$ we can also express the integral in terms of modified Bessel functions.

\begin{lemma}\label{Le1}
We have
\[
I(\mu_+,\mu_-) :=\int_{A > 0} \frac{(dA)}{(\det A)^{3/2}}  
e^{- \frac12  {\rm Tr} \Big (A+ A^{-1}
\Big[\begin{smallmatrix} \mu_+^2 & 0 \\ 0 & \mu_-^2 \end{smallmatrix}\Big] \Big )} 
 = 8 \int_1^\infty \frac{s}{(s^2 - 1)^{1/2}} K_0(s \mu_+) K_0(s \mu_-) \, ds.
\]
\end{lemma}

\begin{proof}
Write for the $2 \times 2$ positive definite matrix $A$
\[
A = \begin{bmatrix} b_1 & c \\
c & b_2 \end{bmatrix}.
\]
Then $A > 0$ is equivalent to
\[
b_1, b_2 > 0 \qquad {\rm and} \qquad b_1 b_2 - c^2 > 0.
\]
Using the notation $h =  b_1 b_2 - c^2$ for the determinant, and expressing this equation
as a delta function constraint allows us to write
\[
I(\mu_+,\mu_-)  = \frac{1 }{ 2 \pi}
\int_0^\infty db_1  \int_0^\infty db_2   \int_{-\infty}^\infty dc \int_{-\infty}^\infty dw
\int_0^\infty
\frac{dh }{ h^{3/2}}
e^{- \frac12 (b_1 + b_2) - \frac1{2h} (b_1 \mu_+^2 + b_2 \mu_-^2) }e^{i w(h - (b_1 b_2 - c^2))}.
\]
The working now is elementary. We first integral over $c$, change variables $h \mapsto b_1 b_2 h$, $w \mapsto w/b_1 b_2$,
and integrate over $w$, then $b_1$ and $b_2$, using the fact that
\[
\int_0^\infty e^{- \frac12 b - \frac12 \frac{\mu^2 }{ hb}} \, \frac{db }{ b}
= 2 K_0\Big ( \frac{\mu }{ \sqrt{h}} \Big ).
\]
The last step is to change variables $s = 1/\sqrt{h}$.
\end{proof}

Alternative expressions for $I(\mu_+,\mu_-)$ are known. One, which involves not the $K_0$ Bessel
function but rather the $I_0$ Bessel function is based on changing variables to the eigenvalues and
eigenvectors of $A$, $B$, and using the matrix integration formula for the integral
over Haar measure of the $2 \times 2$ orthogonal group restricted to matrices with elements in the
first entry of each column positive,
\[
\frac{1 }{ {\rm vol}  \, O^*(2) }
\int e^{{\rm Tr} \, X O Y O^T} (O^T d O) =
e^{\frac12 (x_1 + x_2)(y_1 + y_2)} I_0(\tau), \qquad \tau = - \frac{(x_1 - x_2)(y_1 - y_2) }{ 2},
\]
implying that \cite{IK14}
\[
I(\mu_+,\mu_-)  = \frac{\pi }{ 2}
\int_0^\infty da_1  \int_0^\infty da_2 \,
\frac{|a_1 - a_2| }{ (a_1 a_2)^{3/2}}
e^{- (a_1 + a_2) }e^{- \frac14 ( \frac1{a_1} + \frac1{a_2}) ( \mu_+^2 + \mu_-^2)}
I_0 \Big ( \frac14 \Big ( \frac1{a_1} - \frac1{a_2} \Big ) (\mu_+^2 - \mu_-^2) \Big ).
\]
Another, which is based on working similar to that used in the proof of Lemma \ref{Le1}, but starting
from (\ref{L3b}) rather than (\ref{L3c}) tells us that \cite{APS10}
\begin{equation}\label{2.17}
I(\mu_+,\mu_-)  = 4 \sqrt{\pi}
\int_0^\infty \frac{1}{\sqrt{t}} \exp \Big ( - (\mu_+^2 + \mu_-^2) t - \frac{1}{4t} \Big ) K_0(2 \mu_+ \mu_- t) \, dt.
\end{equation}
There is some advantage in the form (\ref{2.17}), due to its functional dependence on $\mu_+^2 + \mu_-^2$
and $\mu_+\mu_-$, which according to (\ref{L3a}) are given in terms of $\delta,x,y$ by
\[
\mu_+^2 + \mu_-^2 = \delta^2 + 2 (x^2 + y^2), \qquad \mu_+ \mu_- = x^2 + y^2.
\]
Recalling the definition of $I(\mu_+,\mu_-)$ in (\ref{L3c}), this tells us that
\[
W(G) = \frac{1 }{ 2\sqrt{\pi} {\rm vol} \, O(2)}
\int_0^\infty \frac{1 }{ \sqrt{t}} \exp \Big ( -  (\delta^2 + 2 (x^2 + y^2))t - \frac{1}{4t} \Big ) K_0(2(x^2+y^2)t) \, dt.
\]
Substituting in (\ref{L3a}) and using the integral in (\ref{wc1}) to integrate over $\delta$ we obtain
\cite{APS10}
\begin{equation}\label{L3d}
w_c(x,y) = 4
\int_0^\infty \frac{1 }{ t} \exp \Big ( -  2 (x^2 - y^2)t - \frac{1 }{ 4t} \Big ) K_0(2(x^2+y^2)t) \, {\rm erfc} (2\sqrt{t} y) \,dt.
\end{equation}
In the following, we will see that it is possible to calculate the probability finding exactly $k$ eigenvalues without such explicit knowledge of the two-point weight function~\eqref{L3a}. 
Here, we make note of them to make contact with the existing literature and as a reference for a comment in section~\ref{sec:c-eigenvalues}.

Finally, we note that an important difference compared to the result presented in~\cite[Prop. 4.26]{Ip15thesis} is the shift from the two-by-two matrix weight function~\eqref{L3} to~\eqref{L3a} which will be essential in the remaining sections. 

\section{Generalised partition function, skew-orthogonal polynomials and proof of Theorem~1}
\label{sec:skew}

\subsection{Generalised partition function}

Let us denote the joint PDF (\ref{3.1a}) by $\mathcal Q(P_m)$, and define the
generalised partition function for $k$ real and $(N-k)/2$ complex conjugate pairs of eigenvalues by
\begin{equation}\label{L3c1}
Z_{k,(N-k)/2}[u,v] =
\prod_{j=1}^k\int_\R d\lambda_j \, u(\lambda_j)
\prod_{l=1}^{(N-k)/2} \int_{\R\times\R_+} dx_l dy_l  \,  v(x_l,y_l) 
\mathcal Q(P_m).
\end{equation}
We have that with $u = v = 1$ the generalised partition function~\eqref{L3c1} is the probability of finding $k$ real eigenvalues and $(N-k)/2$ complex conjugate pairs of eigenvalues. Functional differentiation of 
\begin{equation}\label{L3d1}
Z_N[u,v] := \sum_{k=0}^N Z_{k,(N-k)/2}[u,v],
\end{equation}
where the sum is restricted to $k$ of the same parity of $N$ allows the correlation functions to be computed;
see e.g.~\cite[\S 15.10]{Fo10}.

Independent of the specific functional form of $w_r$ and $w_c$ in (\ref{3.1a}), an observation of
Sinclair \cite{Si06} tells us that due to the product of difference $\Delta$, the method of integration over
alternative variables implies that $Z_{k,(N-k)/2}[u,v] $ can be written as a Pfaffian.
The details of the necessary working can be found in e.g.~\cite[Prop.~15.10.3, $N$ even]{Fo10} and
\cite[\S 4.3.1 ($N$ even) and \S 4.3.2 $N$ odd)]{Ma11}. We report the final result only.

\begin{prop}\label{U1}
Let $\{p_{l-1}(x) \}_{l=1,\dots,N}$ be a set of monic polynomials, with $p_{l-1}(x)$ of degree $l-1$. Let
\begin{align}\label{be1}
\alpha_{j,k} & = \int_{-\infty}^\infty dx \, u(x)  w_r(x) \int_{-\infty}^\infty dy \, u(y)  w_r(y)
p_{j-1}(x) p_{k-1}(y) {\rm sgn}(y - x) \nonumber \\
\beta_{j,k} & = 2i \int_{\R\times\R_+} dx dy \, v(x,y) w_c(x,y)
\Big ( p_{j-1}(x+iy) p_{k-1}(x - i y) -   p_{k-1}(x+iy) p_{j-1}(x - i y) \Big ),
\end{align}
and
\begin{equation}
\mu_k := \int_{-\infty}^\infty w_r(x) u(x) p_{k-1}(x) \, dx.
\end{equation}
For $k,N$ even we have
\begin{equation}\label{3.33}
Z_{k,(N-k)/2}[u,v] =  \Big ( \frac1{Z_{k,N}} \Big )^m
[\zeta^{k/2}] {\rm Pf} [ \zeta \alpha_{j,l} + \beta_{j,l}]_{j,l=1,\dots,N},
\end{equation}
while for $k,N$ odd we have
\begin{equation}
Z_{k,(N-k)/2}[u,v] =  \Big ( \frac1{Z_{k,N}} \Big )^m
[\zeta^{(k-1)/2}] {\rm Pf} 
\begin{bmatrix}
[ \zeta \alpha_{j,l} + \beta_{j,l}] & [\mu_j]  \\
[-\mu_l] & 0 
\end{bmatrix} ]_{j,l=1,\dots,N}
\end{equation}
with $[\xi^k]f(\xi)$ defined as in Theorem~\ref{T1} and $Z_{k,N}$ given by~\eqref{normalisation}.
\end{prop}

\subsection{Skew orthogonal polynomials}

The matrix $  [ \zeta \alpha_{j,l} + \beta_{j,l}]$ is antisymmetric. For $\zeta = 1$ and $u=v=1$, it is possible
to choose the monic polynomials $\{p_{l-1}(x) \}_{l=1,\dots,N}$ so that this anti-symmetric matrix  is 
block diagonal, with the blocks $2 \times 2$ anti-symmetric matrices
\[
\begin{bmatrix} 0 & h_{j-1} \\
-h_{j -1}& 0 \end{bmatrix},
\]
$j=1,\dots,N/2$, $N$ even, and $j=1,\dots,(N-1)/2$, $N$ odd, with the last diagonal entry 0 in this latter case.
In fact, from a theoretical perspective this is also true for general $\zeta$, however our method for these polynomials (given below) is only valid if $\zeta=1$.
The use of skew-orthogonal polynomials is standard in random matrix theory; see e.g.~\cite[Ch.~6]{Fo10}.
A Gram--Schmidt procedure shows that the construction of such polynomials is always possible, and that they
are unique up to the mapping
\[
p_{2m+1}(x) \mapsto p_{2m+1}(x) + \gamma_{2m} p_{2m}(x).
\]
This mapping, for $\gamma_{2m}$ an arbitrary constant, leaves the skew-orthogonality property unchanged.

With $\alpha_{j,l}$ and $\beta_{j,l}$ specified by~\eqref{be1} define the skew-product 
\begin{equation}\label{Km}
\langle p_j, p_l \rangle :=(\alpha_{j,l} + \beta_{j,l}) |_{u=v=1}.
\end{equation}
In the case $m=1$, when the underlying
point eigenvalues PDF is given by (\ref{3.1}), the corresponding skew-orthogonal polynomials were first determined by
Forrester and Nagao \cite{FN07}. They were found to be
\begin{equation}\label{skew-poly_m=1}
p_{2j}(x) = x^{2j}, \qquad p_{2j+1}(x) = x^{2j+1} - 2j x^{2j-1}
\end{equation}
with normalisation
\begin{equation}\label{skew-norm_m=1}
h_{j-1} := (\alpha_{2j-1,2j} + \beta_{2j-1,2j}) |_{u=v=1}  = 2 \sqrt{2 \pi} \Gamma(2j-1).
\end{equation}
The case $m=2$ has been considered by Akemann and collaborators \cite{APS10,AKP10}.
In fact these authors determined the skew-orthogonal polynomials for a more general model, in which
the matrices $X_1$ and $X_2$ in (\ref{YX}) with $m=2$ are
a general linear combination of Gaussian symmetric and anti-symmetric
matrices, as already noted below Theorem~\ref{T1}. Specialising to the case that
$X_1$ and $X_2$ are both standard Gaussian matrices, we read off the skew-orthogonal polynomials
\begin{equation}\label{skew-poly_m=2}
p_{2j}(x) = x^{2j}, \qquad p_{2j+1}(x) = x^{2j+1} - (2j)^2 x^{2j-1}
\end{equation}
with normalisation
\begin{equation}\label{skew-norm_m=2}
h_{j-1} := (\alpha_{2j-1,2j} + \beta_{2j-1,2j}) |_{u=v=1}  = (2 \sqrt{2 \pi} \Gamma(2j-1))^2.
\end{equation}
Comparing the skew-orthogonal polynomials~\eqref{skew-poly_m=1} and~\eqref{skew-poly_m=2} as well as the normalisations~\eqref{skew-norm_m=1} and~\eqref{skew-norm_m=2} a simple pattern seems apparent: the coefficient in the odd skew-orthogonal polynomials as well as the normalisations constants are raised to powers of $m$, $m=1$ and $m=2$, respectively. This pattern indeed persists in the general case.

\begin{prop}\label{prop:skew}
For the skew-product~\eqref{Km} and $m\in\Z_+$, the polynomials
\begin{equation}\label{1m}
p_{2j}(x) = x^{2j}, \qquad p_{2j+1}(x) = x^{2j+1} - (2j)^m x^{2j-1}
\end{equation}
form a skew-orthogonal set with normalisation
\begin{equation}\label{2m}
h_{j-1} := (\alpha_{2j-1,2j} + \beta_{2j-1,2j}) |_{u=v=1}  = (2 \sqrt{2 \pi} \Gamma(2j-1))^m.
\end{equation}
\end{prop}

In \cite{FN07} the method used to find the skew-orthogonal polynomials was to first establish that with $m=1$
\begin{equation}
\langle x^{2j+1}, x^{2k} \rangle =  \left \{
\begin{array}{ll} - 2^{j+k+3/2} j! \Gamma(k+1/2), & j \ge k, \\
0, & j < k, \end{array} \right.
\end{equation}
which in turn made essential use of knowledge of the explicit functional form of $w_r(x)$ and
$w_c(x,y)$. Some of the details of the working are given in \cite[Proof of Prop.~15.10.4]{Fo10}.
Soon after Sommers \cite{Som07} noted that knowledge of the functional form of the averages of
the product of two characteristic polynomials $C_N(z) = \prod_{j=1}^k (z - \lambda_j) \prod_{s=k+1}^{(N+k)/2}
(z - (x_s + i y_s))(z - (x_s - i y_s))$, summed over $k$ contains sufficient information to fully determine
the skew-orthogonal polynomials. Subsequently, Akemann, Kieburg and Phillips~\cite[Eqns.~(4.6)--(4.7)]{AKP10}
gave the explicit matrix averages formulas
\begin{align}\label{F1}
p_{2n}(z) & = \langle \det(z \mathbb I_{2n} - G) \rangle_{G } \nonumber \\
p_{2n+1}(z) & = z p_{2n}(z) + \langle \det(z \mathbb I_{2n} - G) {\rm Tr} \, G \rangle_{G } ,
\end{align}
for the skew-orthogonal polynomials,
where in the present setting the average over $G$ is over the $m$ standard Gaussian matrices
$X_1,\dots,X_m$ of size
$2n \times 2n$.

Forrester \cite{Fo13a} gave a systematic way to compute averages of the form (\ref{F1}) in the cases
that $G$ is drawn from an ensemble invariant under real orthogonal transformations. This method
was based on the use of zonal polynomials, and was built on ideas contained in \cite{FR09}.
Here we will show that elementary methods suffice to evaluate (\ref{F1}).

 
\begin{proof}[Proof of Proposition~\ref{prop:skew}]
 According to (\ref{YX}), $P_m$ is the product of $m$ independent standard Gaussian matrices $X_1,\dots,X_m$.
 Moreover, from the rule for matrix multiplication, and this specification of the $X_i$, we see that elements taken from distinct rows
$j_1,\dots,j_r$ and columns $k_1,\dots,k_r$, each $j_\mu \ne k_\nu$ are uncorrelated, so that with $P_m = [y_{jk}]_{j,k=1,\dots,N}$
 \begin{equation}\label{F2}
 \Big \langle \prod_{j=1}^r y_{j_r,k_r} \Big \rangle_{X_1,\ldots,X_m} = 0.
\end{equation}
 From the definition of a determinant we have
\begin{equation}\label{F1b} 
  \det(z \mathbb I_{2n} - P_m)  = \sum_{\sigma \in S_{2n}} \varepsilon(\sigma)
  \prod_{l=1}^{2n} (z\delta_{l,\sigma(l)} - y_{l,\sigma(l)}),
 \end{equation}
  where $\varepsilon(\sigma)$ denotes the parity of the 
  permutation $\sigma$ and $\delta_{l,\sigma(l)}  $ denotes the Kronecker delta. Averaging over $X_1,\dots,X_m$ using (\ref{F2})
  shows that the only non-zero term comes from the identity permutation
  and furthermore this average is equal to $z^{2n} $. This establishes $p_{2j}(z)$ in~\eqref{1m}.
  
  For the odd polynomials, multiplication by ${\rm Tr} \, P_m = \sum_{l=1}^{2n} y_{l,l}$ shows that a non-zero
  value will appear after averaging when there is a single monomial $y_{l,l}$ in the expansion of
  $  \det(z \mathbb I_{2n} - P_m) $. We see from (\ref{F1b}) that this is only possible in
  the case of the identity permutation, and that we require the
  coefficient of $z^{2n-1}$ therein. Thus
  \[
 \langle \det(z \mathbb I_{2n} - P_m) {\rm Tr} \, P_m \rangle_{X_1,\ldots,X_m} =
 - z^{2n-1} \sum_{l=1}^{2n} \langle y_{l,l}^2 \rangle_{X_1,\dots,X_m} = - z^{2n-1} (2n)^m.
 \]
 Here the final equality follows by noting that $y_{l,l}$ consists of a sum of $(2n)^{m-1}$ terms which are
 monomials in the elements of the $X_i$, and due to (\ref{F2}) the only terms that survives this averaging
 after squaring are the $(2n)^{m-1}$ perfect squares, which contribute unity.
  Substituting this result into~\eqref{F1} establishes $p_{2j+1}(z)$ in~\eqref{1m}. 
 
 It remains to establish (\ref{2m}). On this point, we first note that from the meaning of $Z_{k,(N-k)/2}[u,v]|_{u=v=1}$
as the probability that there are exactly $k$ real eigenvalues, it follows that $Z_N[u,v]|_{u,v=1} = 1$,
where $Z_N[u,v]$ is specified by (\ref{L3d1}). On the other hand, it follows from (\ref{3.33}) that for $N,k$ even
\begin{equation}
Z_N[u,v] =  \Big ( \frac{1 }{ 2^{N(N+1)/4} \prod_{l=1}^N \Gamma(l/2)} \Big )^m
{\rm Pf} [  \alpha_{j,l} + \beta_{j,l}]_{j,l=1,\dots,N}.
\end{equation}
Setting $u=v=1$, and using the skew-orthogonal polynomials, the RHS can be evaluated 
to give
\[
1 =  \Big ( \frac{1 }{ 2^{N(N+1)/4} \prod_{l=1}^N \Gamma(l/2)} \Big )^m \prod_{l=1}^{N/2} h_{l-1},
\]
and (\ref{2m}) follows.
\end{proof}

\begin{remark}\label{mario-remark}
Examination of the above proof shows that invariance of a single matrix entry under the reflection $y_{jk}\mapsto -y_{jk}$ implies
\[
p_{2n}(z)=z^{2n}
\qquad\text{and}\qquad
p_{2n+1}(z)=z^{2n+1}-\langle \tr P_m^2 \rangle z^{2n-1}.
\]
\end{remark}

\subsection{Probability of \textit{k} real eigenvalues}
\label{sec:k-real}

It has already been remarked below the definition of the generalised partition function~\eqref{L3c1} that with $u=v=1$ this quantity can be interpreted as the probability $p_{N,k}^{P_m}$ that for the ensemble of matrices
specified by (\ref{YX}), with each $X_i$ therein an $N \times N$ real standard Gaussian,
there are exactly $k$ real eigenvalues. This assumes $k$ and $N$ have the same parity; if not the probability
is zero. According to Proposition \ref{U1} these probabilities can be written as Pfaffians. Let us suppose the
polynomials therein are furthermore even (odd) when there degree is even (odd). We then know, by the symmetry
of the integrands, that each $(\zeta \alpha_{j,l} + \beta_{j,l})|_{u=v=1} = 0$ unless the parity of $j$ and $l$ is opposite. Furthermore
making use of the fact that the $(\zeta \alpha_{j,l} + \beta_{j,l})|_{u=v=1}$ is anti-symmetric in $j,l$ allows the
Pfaffian to be written as a determinant of half the size, telling us that for $N,k$ even
\begin{equation}\label{s5}
p_{N,k}^{P_m} =  
 \bigg ( \frac{2^{-N(N+1)/4}}{ \prod_{l=1}^N \Gamma(l/2)} \bigg )^m
[\zeta^{k/2}] \det  [ (\zeta \alpha_{2j-1,2l} + \beta_{2j-1,2l} )|_{u=v=1} ]_{j,l=1,\dots,N/2},
\end{equation}
and for $N,k$ odd
\begin{equation}\label{s6}
p_{N,k}^{P_m} =  
\bigg ( \frac{2^{-N(N+1)/4}}{ \prod_{l=1}^N \Gamma(l/2)} \bigg )^m
[\zeta^{(k-1)/2}]  \det\Big[ (\zeta \alpha_{2j-1,2l} + \beta_{2j-1,2l})|_{u=v=1}] \quad [\mu_{2j-1}]  
\Big]_{\substack{j=1,\dots,(N+1)/2 \\ l=1,\dots,(N-1)/2}}.
\end{equation}
We are now well placed to establish (\ref{a1}) and (\ref{14a}).

\begin{proof}[Proof of Theorem \ref{T1}.] We choose the polynomials $\{p_j(x) \}$ as the skew-orthogonal polynomials (\ref{1m})
so we have
\begin{equation}\label{s0}
(\zeta \alpha_{2j-1,2l} + \beta_{2j-1,2l})|_{u=v=1} = (\zeta - 1) \alpha_{2j-1,2l} |_{u=v=1} + h_{j-1} \delta_{j,l},
\end{equation}
with the explicit value of $u_{j-1}$ being given by (\ref{2m}). Thus we have been able to eliminate the
dependence on $\beta_{j,k}$, which from the definition (\ref{be1}) involves the 
weight $w_c(x,y)$ --- a quantity which from (\ref{L3a}) is not in general known in terms of explicit special functions.
The remaining quantity $ \alpha_{2j-1,2l}$ is specified by (\ref{be1}), and the weight therein $w_r(x)$ is given
as a Meijer $G$-function according to (\ref{C1a}). In fact this very same quantity, up to a proportionality has appeared in the earlier study 
\cite[Proposition 3]{Fo14} and we read off the evaluation
\begin{equation}\label{s1}
 \alpha_{2j-1,2k} |_{u=v=1} 
 = 2^{(j+k-1/2)m} a_{j,k}-2^{(j+k-3/2)m}a_{j,k-1},
 \qquad k,l=1,2,\ldots
\end{equation}
where we use the definition (\ref{aa}) with $a_{j,0}:=0$ since the lowest order odd skew-orthogonal polynomial is a monomial. Substituting (\ref{s1}) in (\ref{s0}), and substituting
the result in turn in (\ref{s5}) we obtain after minor manipulation the formula (\ref{a1}).

To deduce (\ref{14a}) we require the additional evaluation, also contained in  \cite[Proposition 3]{Fo14},
$\mu_{2j-1}=\Gamma(j-{1/2})^m$, and similarly substitute in (\ref{s6}).
\end{proof}


For $m=1$ the probabilities $p_{N,k}^{P_1}$ have been known since the late nineties and they are all of the form $r + \sqrt{2} s$ where $r$ and $s$ are rational numbers \cite{Ed97}. Tabulations can be found in \cite[Table 5]{Ed97} and \cite[Table~2]{AK07}.
Recently, an evaluation of the Meijer $G$-function
\[
 \MeijerG{3}{2}{3}{3}{a_1,a_2,c}{b_1, b_2,c+n}{z},\qquad n\in\N
\]
as a summation over a linear combination of $\{ {}_2F_1(\mu+a,\mu+b;\mu+c;1-z) \}_{\mu=0}^n$ has
been given by Kumar \cite{Ku15}, and this was used to show
\begin{equation}\label{G=?}
\MeijerG{3}{2}{3}{3}{3/2-j,3/2-j,1}{0,k,k}{1} =
\pi^2 \frac{\Gamma(k) \Gamma^2(2j+2k-1)}{\Gamma^2(j+k)}
\sum_{\mu=0}^{k-1} \frac{16^{2-\mu-2j-k}\Gamma^2(2\mu+2j-1)}{\Gamma(\mu+1)\Gamma^2(\mu+j)\Gamma(\mu+2j+k-1)},
\end{equation}
which allows us to get explicit expressions for the probabilities $p_{N,k}^{P_2}$ (i.e. $m=2$). 
Note in particular that this is of the form $\pi^2$ times a rational number, a feature which was conjectured
in \cite{Fo14}. Substituting in (\ref{a1}) and (\ref{14a}) in the case $m=2$ makes the structure of the
probabilities explicit for $p_{N,k}^{P_2}$. These are all polynomials of degree $\lfloor N/2 \rfloor$ in $\pi$ with rational coefficients; probabilities for low values of $N$ are tabulated in Table~\ref{table:k-real}. It is worth noting that similar probabilities for the real spherical and the truncated orthogonal ensembles are also given as polynomials in $\pi$ and $1/\pi$; see~\cite{Ma11} and references therein for an extensive summary.

Beyond the cases $m=1$ and $m=2$, evaluation formulas for the Meijer $G$-function in (\ref{aa}) are challenging. In addition to the contour integral representation~\eqref{GG}, we may also write the Meijer $G$-function as an $m$-fold integral on the real line,
\[
\frac{\MeijerG{m+1}{m}{m+1}{m+1}{3/2-j,\ldots,3/2-j,1}{0,k,\ldots,k}{1}}{\Gamma(j+k-1/2)^m}
=\int\limits_1^\infty\frac{dx_m}{x_m}\prod_{\ell=1}^{m-1}
\bigg[\int\limits_0^\infty\frac{dx_\ell}{x_\ell}\frac{(x_\ell/x_{\ell+1})^{j-1/2}}{(1+x_\ell/x_{\ell+1})^{j+k-1/2}}\bigg]
\frac{x_1^k}{(1+x_1)^{j+k-1/2}},
\]
which may be checked to agree with~\eqref{G=?} for $m=2$. 
Such $m$-fold integral representations give a relation to product of random scalars. However, explicit expressions in terms of elementary functions remain unknown for $m\geq3$.

With an explicit method for calculating the probability of finding $k$ real eigenvalues, it seems natural to ask for different types of number statistics. A prime example would be the expected number of real eigenvalues. Albeit such expectation values may be calculated using Theorem~\ref{T1}, we will see in section~\ref{sec:r-eigenvalues} that the spectral density for the eigenvalues can be used to obtain a more efficient formula. The interest in real eigenvalue statistics, of course, extends beyond the expected number of real eigenvalues. Another common question is to ask for extreme value statistics, i.e. the probability that there are abnormally many (or few) real eigenvalues. As mentioned in the introduction, the probability that all eigenvalues are real has already be studied in~\cite{Fo14}, which led to the remarkable conclusion that this probability tends to unity for $m\to\infty$. It is more challenging to ask for the probability of finding only a few real eigenvalues in the large-$N$ limit, say the probability that an even dimensional product matrix has no real eigenvalues. 

A step in this direction was taken in~\cite{Fo15}, where using the relation to the Brownian annihilation process $A+A\to\varnothing$, the first two terms of the large $s$  asymptotics of the probability that there are no real eigenvalues in an interval of size $s$ near the origin for $N\to\infty$ real Ginibre ($m=1$) was computed. It was realized Kanzieper et al.~\cite{KPTTZ15} that heuristic at least this result implies for large $N$
\begin{equation}
\frac1{\sqrt{N}}\log p_{N,0}^{P_1}=-\frac{1}{\sqrt{2\pi}}\zeta\Big(\frac32\Big)
+\frac C{\sqrt N}+\cdots
\end{equation}
with  $\zeta(x)$ denotes the Riemann zeta function and
\begin{equation}
C=\log 2-\frac14+\frac1{4\pi}\sum_{n=2}^\infty\frac1n\Big(-\pi+\sum_{p=1}^{n-1}\frac{1}{p(n-p)}\Big)\approx 0.0627,
\end{equation}
and moreover these authors gave a rigorous proof of the leading term. It is not known how to generalize the workings of~\cite{KPTTZ15}, which are based on Theorem~\ref{T1}, beyond $m=1$. However, our Theorem~\ref{T1} at least allows us to establish numerical estimates, e.g. fitting $aN^{1/2}+bN^0+cN^{-1/2}$ to $\log p_{N,0}^{P_2}$ for $N=50,52,\ldots,120$ suggest that
\begin{equation}
\lim_{N\to\infty}\frac1{\sqrt{N}}\log p_{N,0}^{P_2}\approx -1.474
\end{equation}
for $N$ even. We note that $1.474>\zeta(3/2)/\sqrt{2\pi}\approx 1.042$, which is in the agreement with the expectation that $p_{N,0}^{P_m}$ decreases with increasing $m$.

\begin{table}
\caption{Probabilities, $p_{N,k}^{P_2}$, of finding $k$ real eigenvalues given a product of two $N\times N$ real Ginibre matrices for small $N$ and $k$. The tabulated values are (left to right) the exact expression for the probability, the decimal expansion of the exact value, and the numerical values obtained from a simulation with one million realizations of the matrix product. For comparison, the rightmost column present the equivalent values for the standard Ginibre ensemble (i.e. $m=1$).}\label{table:k-real}
\[
\arraycolsep=1.4pt\def\arraystretch{1.5}
\begin{array}{@{\quad}c@{\quad}|@{\quad}r@{\qquad}c@{\qquad}c@{\qquad}c@{\quad}}
 & \text{Exact} & \text{Approx.} & \text{Simul.} & (P_{N,k}^{P_1})\\
\hline\hline 
p_{2,0}^{P_2} & 1-\frac14\pi & 0.2146 & 0.2144 & (0.2929)\\
p_{2,2}^{P_2} & \frac14\pi & 0.7854 & 0.7856 & (0.7071) \\
\hline 
p_{3,1}^{P_2} & 1-\frac{5}{32}\pi & 0.5091 & 0.5091 & (0.6464) \\
p_{3,3}^{P_2} & \frac{5}{32}\pi & 0.4909 & 0.4909  & (0.3536) \\
\hline
p_{4,0}^{P_2} & 1-\frac{755}{2048}\pi+\frac{201}{8192}\pi^2 & 0.0840 & 0.0841 & (0.1527) \\
p_{4,2}^{P_2} & \frac{755}{2048}\pi-\frac{201}{4096}\pi^2 & 0.6738 & 0.6746 & (0.7223) \\
p_{4,2}^{P_2} & \frac{201}{8192}\pi^2 & 0.2422 & 0.2413 & (0.1250) \\
\hline
p_{5,1}^{P_2} & 1-\frac{4185}{16384}\pi+\frac{10013}{1048576}\pi^2 & 0.2918 & 0.2922 & (0.4567) \\
p_{5,3}^{P_2} & \frac{4185}{16384}\pi-\frac{10013}{524288}\pi^2 & 0.6140 & 0.6137 & (0.5120)\\
p_{5,5}^{P_2} & \frac{10013}{1048576}\pi^2 & 0.0942 & 0.0942 & (0.0313) \\
\hline
p_{6,0}^{P_2} & 1-\frac{3821355}{8388608}\pi+\frac{873624317}{17179869184}\pi^2-\frac{64011585}{68719476736}\pi^3 & 0.0419 & 0.0420 &(0.0935)\\
p_{6,2}^{P_2} & \frac{3821355}{8388608}\pi-\frac{873624317}{8589934592}\pi^2+\frac{192034755}{68719476736}\pi^3 & 0.5140 & 0.5139 & (0.6529) \\
p_{6,4}^{P_2} & \frac{873624317}{17179869184}\pi^2-\frac{192034755}{68719476736}\pi^3 & 0.4152 & 0.4154 & (0.2481) \\
p_{6,6}^{P_2} & \frac{64011585}{68719476736}\pi^3 & 0.0289 & 0.0287 & (0.0055) \\
\hline
p_{7,1}^{P_2} & 1-\frac{22392747}{67108864}\pi+\frac{105417740207}{4398046511104}\pi^2-\frac{31625532537}{140737488355328}\pi^3 &
0.1813 & 0.1817 & (0.3374) \\
p_{7,3}^{P_2} & \frac{22392747}{67108864}\pi-\frac{105417740207}{2199023255552}\pi^2+\frac{94876597611}{140737488355328}\pi^3 &
0.5960 & 0.5959 & (0.6639) \\
p_{7,5}^{P_2} & \frac{105417740207}{4398046511104}\pi^2-\frac{94876597611}{140737488355328}\pi^3 & 0.2157 & 0.2154 & (0.0846)\\
p_{7,7}^{P_2} & \frac{31625532537}{140737488355328}\pi^3 & 0.0070 & 0.0070 & (0.0007)
\end{array}
\]
\end{table}

\section{Correlation functions}
\label{sec:corr}


The Pfaffian formulae of Proposition \ref{U1} for the generalised partition function,
combined with the simplification inherent in the use of skew-orthogonal polynomials,
$\{ p_j(x) \}$, allow the $k$-point correlation to be expressed in the form (\ref{K0}) with entries given
in terms of $\{ p_j(x) \}$. While (\ref{K0}) refers to the real-to-real eigenvalue correlations, this same structure
remains true for the general correlation functions. In fact, the entries of
the correlation kernel also have the same structure; see e.g.~\cite[\S 4.5 and \S 4.6]{Ma11}.

Let $\mathbb C_+=\{x+iy\in\C:y>0\}$ denote the upper complex half-plane, and specify $w_r(x)$ and $w_c(x,y)$
according to (\ref{C1a}) and (\ref{L3a}). Define
\begin{align}\label{AB}
q_j(\mu) &= \left \{
\begin{array}{ll} 
w_r(x) p_j(x), & \mu = x \in \mathbb R \\
(\frac{1}{2} w_c(x,y))^{1/2} p_j(x+iy), & \mu = x + i y \in \mathbb C_+ 
\end{array}   \right. \nonumber \\[.5em]
\tau_j(\mu) &= \left \{
\begin{array}{ll} 
\displaystyle - \frac12 \int_{-\infty}^\infty {\rm sgn} \, (x-y) q_j(y) \, dy & \mu = x \in \mathbb R \\
i q_j(x+iy),  & \mu = x + i y \in \mathbb C_+
\end{array}   \right. \nonumber \\[.5em]
\epsilon(\mu,\eta) &=  \left \{  
\begin{array}{ll}  \frac12 {\rm sgn} \, (\mu - \eta), & \mu,\eta \in \mathbb R, \\
0, & {\rm otherwise} 
\end{array}   \right.  
\end{align}
In this notation, the entries of the correlation kernel (\ref{K0}) in the case of the correlation between
real eigenvalues only, or the correlation between complex conjugate pairs of eigenvalues are given by
\begin{align}\label{AB1}
S(\mu,\eta) & = 2 \sum_{j=0}^{N/2 - 1} \frac1{u_j} \Big (
q_{2j}(\mu) \tau_{2j+1}(\eta) - q_{2j+1}(\mu) \tau_{2j}(\eta) \Big ), \nonumber \\
D(\mu,\eta) & =
2 \sum_{j=0}^{N/2 - 1} \frac1{u_j} \Big (
q_{2j}(\mu) q_{2j+1}(\eta) - q_{2j+1}(\mu) \tau_{2j}(\eta) \Big ), \nonumber \\
\tilde{I}(\mu,\eta) & =
2 \sum_{j=0}^{N/2 - 1} \frac1{u_j} \Big (
\tau_{2j}(\mu) \tau_{2j+1}(\eta) - \tau_{2j+1}(\mu) \tau_{2j}(\eta) \Big ) + \epsilon(\mu,\eta).
\end{align}
For $N$ odd these expressions require modification; see e.g.~\cite{FM08}, \cite[\S 4.6]{Ma11}. For
efficiency of presentation, we will restrict attention to the $N$ even case.

Our main interest in section~\ref{sec:c-eigenvalues} and~\ref{sec:r-eigenvalues} will be spectral densities (one-point correlation functions) and quantities derivable from these. For this reason, we focus on the complex-to-complex and the real-to-real eigenvalue correlations, but real-to-complex correlations can be treated in a similar manner.

\subsection{Complex eigenvalues}
\label{sec:c-eigenvalues}

We see from (\ref{AB}) and (\ref{AB1}) that in the case of the correlation between complex
eigenvalues, up to factors involving $w_c(x,y)$ all the quantities are polynomials, and are related by
\begin{equation}\label{ID-complex}
\tilde I(w,z) = i S(\bar{w},z), \quad
D(w,z) = - i S(w,\bar{z}).
\end{equation}
Thus it suffices to consider $S(w,z)$, where $w=u+iv$ and $z=x+iy$.
For this, (\ref{AB}) and (\ref{AB1})  tell us that
\[
S(w,z) =  2i(w_c(u,v) w_c(x,y))^{1/2}
\sum_{j=0}^{N/2 - 1} \frac{p_{2j}(w) p_{2j+1}(\bar{z}) - p_{2j+1}(w) p_{2j}(\bar{z})}{h_j}
\]
Upon use of the skew-orthogonal polynomials given by Proposition~\ref{prop:poly} this simplifies to
\begin{equation}
S(w,z) = 2i(w_c(u,v) w_c(x,y))^{1/2}\sum_{j=0}^{N-2}\frac{(\bar{z}-w)(w\bar{z})^j}{(2\sqrt{2\pi}\,j!)^m}.
\end{equation}
We are typically interested in either a global scaling regime (where the eigenvalues are concentrated within a region with compact support) or local scaling regimes (where the eigenvalue interspacing is order unity). For simplicity, let us focus on the one-point function (i.e. the spectral density) which for complex eigenvalues is given by $\rho_{(1)}^c(z)=S(z,z)$.

The global scaling regime for the spectral density is known from free probability \cite{BJW10,GT10,OS11},
\begin{equation}\label{L3eA}
\lim_{N \to \infty} N^{m-1}  \rho_{(1)}^c(N^{m/2}w) = \frac{|w|^{(2/m) - 2} }{ m \pi} \chi(1>|w|),
\end{equation}
where $\chi(A) = 1$ if $A$ is true, $0$ otherwise. This holds because the full spectral density (i.e. including complex as well as real eigenvalues) is dominated by the complex spectrum in the global scaling regime. We note that there also exists a global scaling regime for the real spectrum, albeit sub-dominant. We will return to this limit in section~\ref{sec:r-eigenvalues}.

On the local scale, the region near the origin is of greatest interest since it gives rise to new types scalings (i.e. different than the ordinary Ginibre case). The local density near the origin is given by
\begin{equation}\label{pse}
\lim_{N \to \infty} \rho_{(1)}^c(z) = \frac{2y\,w_c(x,y)}{(2 \sqrt{2 \pi})^m}\MeijerG{1}{0}{0}{m}{-}{0,\ldots,0}{-\abs{z}^2},
\end{equation}
since
\begin{equation}\label{sum-hyper-meijer}
\sum_{j=0}^{\infty} \frac{x^{j} }{ (j!)^m}=\hypergeometric{0}{m-1}{-}{1,\ldots,1}{x}
=\MeijerG{1}{0}{0}{m}{-}{0,\ldots,0}{-x}.
\end{equation}
We recall from section~\ref{sec:jpdf} that the weight function $w_c(x,y)$ has an explicit and concise expression for $m=1,2$ but not for $m>2$. We note that if $m=1,2$ then the Meijer $G$-function in~\eqref{pse} evaluates as
\[
\MeijerG{1}{0}{0}{1}{-}{0}{-\abs{z}^2}=e^{-\abs z^2}
\qquad\text{or}\qquad
\MeijerG{1}{0}{0}{2}{-}{0,0}{-\abs{z}^2}=I_0(2\abs{z})
\]
with the latter being a modified Bessel function. Combining this with the weight functions from section~\ref{sec:jpdf} reproduces known formulae for the density (the $m=2$ case was given in~\cite{APS10}).


\subsection{Real eigenvalues}
\label{sec:r-eigenvalues}

In this section we focus on the part of the spectrum which is located on the real axis. Similarly to the complex spectrum described above, all correlations may be expressed in terms of the pre-kernel $S(x,y)$. We see from (\ref{AB}) and (\ref{AB1}) that
\begin{equation}\label{ID-real}
D(x,y) = - \frac{\partial }{ \partial y} S(x,y), \qquad
\tilde{I}(x,y) = - \int_x^y S(t,y) \, dt + \frac12 {\rm sgn} \, (x-y),
\end{equation}
which produce the correlation functions by insertion in~\eqref{K0}. We note that the relations between the pre-kernels~\eqref{ID-real} are more complicated for the real-to-real correlations than for the complex-to-complex correlations where the pre-kernels are related according to~\eqref{ID-complex}. On the other hand, the weight functions are simpler in the real case~\eqref{C1a} than in the complex case~\eqref{L3a}.

Using (\ref{AB}) and (\ref{AB1}) and the skew-orthogonal polynomials (Proposition~\ref{prop:poly}), we write the pre-kernel as
\begin{equation}\label{sxy}
S(x,y) = \int_{-\infty}^\infty dv\,(x-v)\,\sign(y-v)\, w_r(x)  w_r(v) \sum_{j=0}^{N-2} \frac{(xv)^j }{(j!)^m}.
\end{equation}
For $m=1$ (i.e. the ordinary Ginibre ensemble), the sum may be rewritten as an incomplete gamma function times an exponential and the integral over $v$ can be performed, which yields~\cite{EKS94}
\begin{equation}\label{Sm=1}
S(x,y)\vert_{m=1}=\frac1{\sqrt{2\pi}}
\bigg(
e^{-(x-y)^2/2}\frac{\Gamma(N-1,xy)}{\Gamma(N-1)}+2^{(N-3)/2}e^{-x^2/2}x^{N-1}\sign(y)\frac{\gamma((N-1)/2,y^2/2}{\Gamma(N-1)}
\bigg).
\end{equation}
This formulation of the pre-kernel is extremely useful in the study of large-$N$ asymptotics. Unfortunately there are no direct generalisation of this result to $m\geq 2$, which makes asymptotic analysis more challenging. However, it is possible to perform the integral over $v$ in~\eqref{sxy} for arbitrary $m$. To do so, we rewrite~(\ref{sxy}) as
\begin{equation}\label{sxy1}
S(x,y)=\sum_{j=0}^{N-2} \frac{w_r(x)\,x^j}{(2\sqrt{2\pi}\,j!)^m}(x A_j(y) - A_{j+1}(y)),\qquad
A_j(y):=  \int_{-\infty}^\infty w_r(v)  \sign (y-v) v^j \, dv.
\end{equation}
Now, standard identities for the Meijer $G$-function give
\begin{equation}\label{CAa}
A_j(y) = \left \{ 
\begin{array}{ll} \displaystyle 
- 2^{m(1+j)/2} \MeijerG{m+1}{0}{1}{m+1}{1}{0, (1+j)/2,\dots,(1+j)/2}{\frac{y^2}{2^m}}, & j \ \textup{odd} \\[.4cm]
 \displaystyle  
 y^{1+j} \MeijerG{m}{1}{1}{m+1}{-(j-1)/2}{0,\dots,0,-(j+1)/2}{\frac{y^2}{2^m}}, 
  & j \ \textup{even}. 
\end{array} \right.
\end{equation}
This form of the pre-kernel $S(x,y)$ is useful if we are interested in the expected number of real eigenvalues. We recall that the expected number of real eigenvalues, $\E(\#\text{reals})$, can be found by integration over the real spectral density $\rho^r_{(1)}(x)=S(x,x)$. Thus, it follows from~\eqref{sxy1} that
\begin{equation}\label{N6}
\mathbb E (\#{\rm reals})  = \int_{-\infty}^\infty \rho_{(1)}^r(x) \, dx 
= 2\sum_{j=0}^{N-2} \frac{\alpha_{j+1,j+2}}{(2\sqrt{2\pi}\,j!)^m} ,
\end{equation}
where
\[
\alpha_{j,k} = \int_{-\infty}^\infty dx  \int_{-\infty}^\infty dy \, w_r(x) w_r(y) x^{j-1} y^{k-1} {\rm sgn} \, (y-x).
\]
The quantity $ \alpha_{j,k} $ is precisely the same quantity appearing in the study \cite{Fo13}, which evaluates to
$\alpha_{2j-1,2k} = 2^{(j+k-1/2)m} a_{j,k}$ with $a_{j,k}$ given by the Meijer $G$-function (\ref{aa}). In the case where the first index
of $\alpha_{j,k}$ is even and the second index odd, we use the anti-symmetric property $\alpha_{j,k}=-\alpha_{k,j}$. This gives
\begin{equation}\label{Ereal}
\mathbb E (\#{\rm reals})   =  2\sum_{j=0}^{N-2} (-1)^j\Big(\frac{2^{j}}{\sqrt{\pi}\,j!}\Big)^ma_{\lceil j/2+1\rceil,\lfloor j/2+1\rfloor},
\end{equation}
where $\lceil\cdot\rceil$ and $\lfloor\cdot\rfloor$ denote the ceiling and floor function, respectively. 
We recall that the formulae above assume that $N$ is even (for odd $N$ the expression~\eqref{Ereal} is altered by the addition of unity). As already mentioned, an evaluation of $a_{j,l}$ in terms of arithmetic constants is only known for $m=1,2$; consequently the same holds for~\eqref{Ereal}. The $m=1$ case is known since the mid nineties~\cite{EKS94}, while the $m=2$ case is evaluated using~\eqref{G=?}; the results for small $N$ are tabulated in Table~\ref{table:E(real)}. 
As anticipated, Table~\ref{table:E(real)} reveals that the expected value of real eigenvalues are consistingly larger for $m=2$ than for $m=1$. For $m>2$ a computation of the expectation value~\eqref{Ereal} requires numerical evaluation of the Meijer $G$-functions.
The expected number of real eigenvalues can, of course, also be obtained using the probabilities given by Theorem~\ref{T1}. In fact, for $m=2$ and small $N$ the expected number of real eigenvalues follows immediately from Table~\ref{table:k-real}, e.g. for $N=4$ we see that
\[
\E(\# \text{reals})\vert_{m=2,N=4}=
 0\Big(1-\frac{755}{2048}\pi+\frac{201}{8192}\pi^2\Big)
+2\Big(\frac{755}{2048}\pi-\frac{201}{4192}\pi^2\Big)
+4\Big(\frac{201}{8192}\pi^2\Big)
=\frac{755}{1024}\pi,
\]
which agrees with Table~\ref{table:E(real)}. 

\begin{table}[htbp]
\caption{Expected value of the number of real eigenvalues for a product of two $N\times N$ Ginibre matrices.  The tabulated values are (top to bottom) the exact expression for the probability, the decimal expansion of the exact value, and numerical values obtained from a simulation with one million realizations of the matrix product. The bottom row shows the corresponding expectation values for the Ginibre ensemble for comparison. It is evident that the exact expressions exhibit a special arithmetic pattern.}\label{table:E(real)}
\[
\arraycolsep=1.4pt\def\arraystretch{1.5}
 \begin{array}{@{\quad}c@{\quad}|@{\quad}c@{\qquad}c@{\qquad}c@{\qquad}c@{\qquad}c@{\qquad}c@{\quad}}
 \mathbb E[\#\text{reals}]_{m=2} & N=2 &  N=3 &  N=4 &  N=5 & N=6 & N=7\\
 \hline\hline
\text{Exact} &
\frac12\pi &
1+\frac{5}{16}\pi &
\frac{755}{1024}\pi &
1+\frac{4185}{8192}\pi &
\frac{3821355}{4194304}\pi &
1+\frac{22392747}{33554432}\pi 
\\
\text{Aprox.} & 1.5708 & 1.9817 & 2.3163 & 2.6049 & 2.8622 & 3.0966 \\
\text{Simul.} & 1.5704 & 1.9813 & 2.3168 & 2.6030 & 2.8607 & 3.0948 \\
(\mathbb E[\#\text{reals}]_{m=1}) & (1.4142) & (1.7071) & (1.9445) & (2.1490) & (2.3312) & (2.4971)
\end{array}
\]
\end{table}

Let us return to the pre-kernel~\eqref{sxy} and consider large-$N$ asymptotics for the real spectral density. Similarly to section~\ref{sec:c-eigenvalues} we focus on the local density near the origin and the global density. Using~\eqref{sum-hyper-meijer}, it is immediately seen that the local scaling regime near the origin gives (\ref{r-local}) announced in Theorem \ref{T2}.
Compared to the same result for the complex density~\eqref{pse}, the real density has the advantage that the weight function $w_r(x)$ has a known expression as a Meijer $G$-function~\eqref{C1a} for all $m$ while $w_c(x,y)$ does not. We note that for $m=1$ the Meijer $G$-functions in~\eqref{r-local} are all simple exponentials; this allows integration over $v$ and confirms that the local spectral density is constant for $m=1$. Moreover, for $m=1$ the corresponding $k$-point correlation takes on the
explicit form
\begin{equation}\label{K1}
\mathbf K^{\rm rr} (x,y) = \begin{bmatrix} \displaystyle \frac{1 }{ \sqrt{2 \pi}} (y - x) e^{-(x-y)^2} & 
\frac{1 }{ \sqrt{2 \pi}} e^{-(x-y)^2}  \\
 -   \frac{1 }{ \sqrt{2 \pi}}  e^{-(x-y)^2}  & \frac{1 }{ 2} {\rm sgn} \, (x-y) \, {\rm erfc}(|x-y|/\sqrt{2})
\end{bmatrix},
\end{equation}
as obtained in \cite{FN07,Som07,BS09}. We remark that it has been argued by Beenakker and co-workers \cite{BEDPSW13} that the statistical state implied by
(\ref{K1}) is realised by the level crossings of so-called Majorana zero modes for a disordered 
semiconducting wire at a Josephson junction, 
in a weak magnetic field.  And this same correlation kernel 
appears in the seemingly unrelated problem  of 
the annihilation process  $A + A \to \emptyset$ in the limit $t \to \infty$ \cite{MbA01,TZ11}.

A study of the global scaling regime for the real spectrum is more challenging. Unlike the complex spectral density (section~\ref{sec:c-eigenvalues}), we have no help from free probability. A qualified guess for this spectral density might be obtained by looking at the $m$-th power of a real Ginibre matrix rather than at the product of $m$ independent matrices. It is immediate that the $m$-th power and the $m$-th product share the same complex macroscopic spectral density, thus assuming that this extends to the real spectrum we expect that
\begin{equation}\label{rho-global-real}
\lim_{N\to\infty}\frac{N^{(m-1)/2}\rho^r_{(1)}(N^{m/2}x)}{\E(\#\text{reals})}=\frac{\abs{x}^{(1/m)-1}}{2m}\chi(x^2<1),
\end{equation}
where $\chi(A)$ is defined as in~\eqref{L3eA}. For $m=1$ the density~\eqref{rho-global-real} is well-known~\cite{EKS94}; a verification follows from~\eqref{Sm=1} using known asymptotics for the incomplete gamma functions. Moreover, we see that the real spectrum~\eqref{rho-global-real} develops a non-integrable singularity at the origin when $m$ tends to infinity similarly to~\eqref{L3eA} as we would expect.
For $m\geq2$ we have no rigorous derivation of~\eqref{rho-global-real} but the form~\eqref{rho-global-real} is supported by (i) a heuristic saddle point analysis and (ii) numerical data. 

Let us first look at the saddle point analysis, which takes~\eqref{sxy} as the starting point. The first step is to introduce an approximation for the sum in~\eqref{sxy}. We know from~\cite[Appendix C]{AB12} that
\begin{equation}
 e^{-mNx^{1/m}}\sum_{j=0}^{N-2}\frac{(Nx)^j}{(j!)^m}\approx \frac1{\sqrt{m}}\Big(\frac{2\pi x^{1/m}}N\Big)^{(1-m)/2}
\end{equation}
for $\abs x<1$ while exponentially suppressed in $N$ for $\abs x>1$. An approximation for the weight function is known from the literature on special functions~\cite{Fi72}, and we have
\begin{equation}
w_r(N^{m/2}x)\approx\frac1{\sqrt{m}}\Big(\frac{4\pi x^{2/m}}{N}\Big)^{(m-1)/2}e^{-\frac12mNx^{2/m}}.
\end{equation}
We insert these approximations into~\eqref{sxy} and want to evaluate the integral over $v$ using a saddle point approximation. Note that there are two maxima of the integrand symmetrically distributed around $v=x$ (the integrand is equal to zero at $v=x$). These two maxima tend to $x$ from left or right, respectively, as $N$ tends to infinity. Thus, we will use an ansatz $v_*=x\pm f(x)$ for our saddle points where $f(x)$ is sub-dominant in $N$. With this ansatz and expanding to lowest order, the saddle points are found to be
\[
v_*\approx x\pm\sqrt{\frac mN}\,x^{1-(1/m)}.
\]
Evaluation at either of these saddle points yields the conjectured form~\eqref{rho-global-real} up to a normalisation. 

Finally, let us compare the density~\eqref{rho-global-real} with a simulation of the random matrix product. Figure~\ref{fig:histogram} shows the visual similarity between the density~\eqref{rho-global-real} for $m=2$ and numerical data stemming from a simulation of $1\,000$ matrix products with $N=1024$. It should be noted that convergence is expected to be exponentially fast in the bulk but considerably slower near the edges. Similar numerical tests have been performed for $m=3,4,5$ and it has been verified that the difference between the analytic formula~\eqref{rho-global-real} and the numerical data decreases with increasing $N$. Furthermore, we expect that the real global density~\eqref{rho-global-real} is universal in the sense that the Gaussian entries may be replaced by other independent variables under suitable assumptions on their moments. This type of universality is known to hold for the complex spectra~\cite{GT10,OS11} and the expectation that such results extend the real spectra is strengthend by numerical comparison generated from random sign ($\pm 1$) matrices. Although it seems a very natural problem, this type of universality for the real global spectrum has received little attention in the literature; this is true even for the classical Ginibre ensemble ($m=1$).

\begin{figure}[htbp]
\centering
\includegraphics{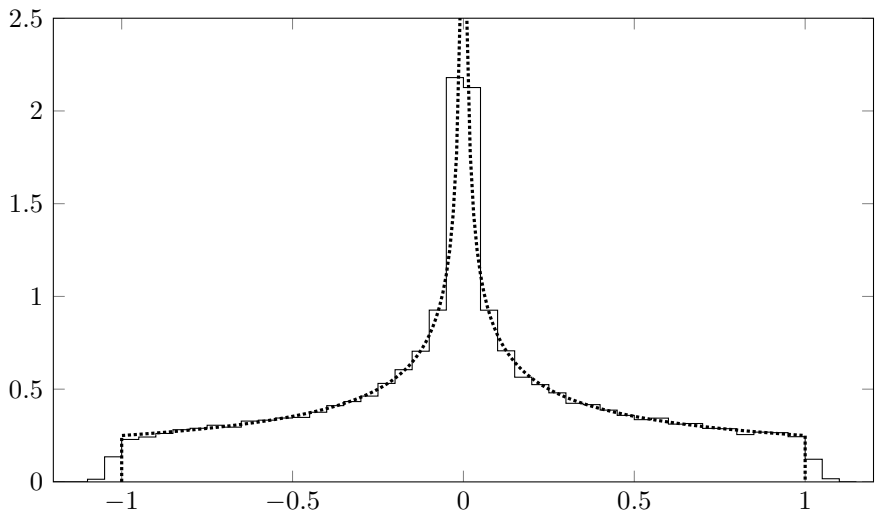}
\caption{The dotted curve shows the global spectral density~\eqref{rho-global-real} for $m=2$, while the histogram shows the distribution of the real eigenvalues from $1\,000$ realisations of a product of two $1024\times 1024$ real asymmetric Gaussian matrices ($36\,390$ eigenvalues in total).}\label{fig:histogram}
\end{figure}

\section{Rectangular matrices}
 
A generalisation to the case of rectangular matrices is also available and we briefly treat it here. The main idea when dealing with a product of random matrices is to reformulate problem as a product of square random matrices with the same eigenvalue properties; this is possible due to a general reduction procedure presented in~\cite{IK14} (see also~\cite[Prop. 2.4]{Ip15thesis}). After this reformulation, the approach is similar to the previous sections because Proposition~\ref{U1} as well as the formulae~\eqref{AB} and~\eqref{AB1} are completely general. Due to this similarity we will only sketch the main ideas here.

We consider a product matrix,
\begin{equation}\label{YXrect}
P_m^\nu=X_1\cdots X_m,
\end{equation}
where each matrix $X_i$ has dimensions $(N+\nu_{i-1})\times(N+\nu_i)$ with $\{\nu_i\}_{i=0,\ldots,m}$ denoting non-negative integers such that $\nu_0=\nu_m=0$. Here the constraint is introduced to ensure that the product matrix is square and has $N$ non-trivial eigenvalues. We note that if $\nu_0=\nu_m>0$ but $\nu_j=0$ for some $0<j<m$ (i.e. the smallest matrix dimension is still $N$) then there will be $\nu_0$ eigenvalues which are trivially equal to zero (and therefore real) but the joint PDF otherwise remains the same except for an obvious change in normalisation. Consequently, all formulae given below may effortlessly be extended to the $\nu_0>0$ case if desired.
 
The generalisation of the probabilities~\eqref{11} with~\eqref{aa} for a purely real spectrum have already appeared in the thesis~\cite[Prop. 4.29]{Ip15thesis}. They are given by
\begin{equation}\label{11rect}
p_{N,N}^{P_m^\nu} = 
\prod_{k=1}^m\prod_{j=1}^N\frac{1}{\Gamma\big(\frac{j+\nu_k}2\big)}\times 
\begin{cases}
\det\Big[ [a_{j,k}^\nu]_{k=1,\ldots,N/2}^{j=1,\ldots,N/2}\Big], & N\ \text{even}\\[2mm]
\det\Big[ [a_{j,k}^\nu]_{k=1,\ldots,(N-1)/2}^{j=1,\ldots,(N+1)/2}\quad 
[\tilde a_j^\nu]_{j=1,\ldots,(N+1)/2}\Big], & N\ \text{odd}
\end{cases}
\end{equation}
with
\begin{equation}\label{aa-rect}
a^\nu_{j,k}:=\MeijerG{m+1}{m}{m+1}{m+1}{\frac{3+\nu_1}{2}-j,\dots,\frac{3+\nu_m}{2}-j,1}{0,k+\frac{\nu_1}2,\dots, k+\frac{\nu_m}2}{1}
\qquad\text{and}\qquad
\tilde a_j^\nu:=\prod_{k=1}^m\Gamma\Big(j + \frac{\nu_k-1}2\Big).
\end{equation}
These formulae allow us to make some straightforward generalisations of the exact expressions presented by Kumar~\cite{Ku15} in the $m=2$ case. Following~\cite{Ku15}, we have 
\begin{equation}\label{meijer-eval}
 \MeijerG{3}{2}{3}{3}{\tfrac32-i,\tfrac32-\tfrac\nu2-i,1}{0,\tfrac\nu2+j,j}{1} 
 =\Gamma(j)\Gamma(i+j-\tfrac12)\Gamma(i+j+\tfrac\nu2-\tfrac12)
 \sum_{k=0}^{j-1}\frac{\Gamma(k+i-\tfrac12)\Gamma(k+i+\tfrac\nu2-\tfrac12)}{\Gamma(k+1)\Gamma(k+2i+j+\nu-1)}.
\end{equation}
The next step would be to rewrite gamma functions with a non-integer argument using Gauss' duplication formula. The right-hand side of~\eqref{meijer-eval} evaluates as $r\pi^2$ for even $\nu$ and $r\pi$ for odd $\nu$ where $r$ denotes some rational constant (depending on both $N$ and $\nu$).  This difference in the power of $\pi$ for even and odd $\nu$ has a remarkable consequence: for even $\nu$ the probabilities~\eqref{11rect} are given as a rational number times $\pi^{\lfloor N/2\rfloor}$ but for odd $\nu$ these constants are simple rational constants (i.e. there is no powers of $\pi$).
The probabilities of a purely real spectrum are tabulated in Table~\ref{table:rect-1} for small values of $N$ and $\nu$.

\begin{table}[htbp]
\caption{Consider a product of two Gaussian matrices, $X_1X_2$, with dimensions $N\times (N+\nu)$ 
and $(N+\nu)\times N$ for $X_1$ and $X_2$, respectively. The table provides exact probabilities for a purely real spectrum for various values of $N$ and $\nu$ (the $\nu=0$ column have previously appeared in~\cite{Ku15}). These probabilities are all given as a rational number times a power of $\pi$ depending on $N$ and whether $\nu$ is odd or even.}\label{table:rect-1}
\[
\arraycolsep=1.4pt\def\arraystretch{1.5}
\begin{array}{@{\quad}c@{\quad}|@{\quad}c@{\qquad\quad}c@{\qquad\quad}c@{\qquad}c@{\quad}}
 p_{N,N}^{P_2^\nu} & \nu=0 & \nu=1& \nu=2 & \nu=3 \\
 \hline\hline
 N=2 & \frac{1}{4}\pi & \frac{1}{2} & \frac1{16}\pi & \frac1{18}  \\
 N=3 & \frac{5}{32}\pi & \frac14 & \frac{7\pi}{256}\pi  & \frac1{45} \\
 N=4 & \frac{201}{8192}\pi^2 & \frac{3}{64} & \frac{233}{524288}\pi^2 & \frac{331}{1512000} \\
 N=5 & \frac{10013}{1048576}\pi^2 & \frac{23}{1152} & \frac{74989}{402653184}\pi^2 & \frac{18461}{190512000} \\
 N=6 & \frac{64011585}{68719476736}\pi^3 & \frac{311}{294912} &  \frac{16981105}{35184372088832}\pi^3 & \frac{41938693}{394314117120000} \\
 N=7 & \frac{31625532537}{140737488355328}\pi^3 & \frac{1349}{5898240}
    & \frac{71615920731}{720575940379279360}\pi^3 & \frac{5963869169}{81934593740800000}
\end{array}
\]
\end{table}

As we have seen in previous sections, to extend the probabilities for a purely real spectrum~\eqref{11rect} to the probabilities $P_{N,k}^{P_m^\nu}$ we need a formula for the joint PDF of the eigenvalues and a formula for the skew-orthogonal polynomials, i.e. generalisations of Theorem~\ref{T3} and Proposition~\ref{prop:skew}. Given such generalisations the rest of the results presented in previous sections may be extended as well due to the generality of Proposition~\ref{U1}.

\begin{prop}\label{prop:jpdf-rect}
Given a Gaussian product matrix~\eqref{YXrect} of dimension $N$ with $k$ real eigenvalues, $\{\lambda_l\}_{l=1}^k$, and $(N-k)/2$ complex conjugate pairs of a eigenvalues, $\{x_j\pm iy_j\}_{j=1}^{(N-k)/2}$, the joint PDF for these eigenvalues is given by
\begin{equation}
\frac{1}{ k! ((N-k)/2)! } \frac1{Z_{N}^{m,\nu}}
\abs[\Big]{ \Delta\Big(\{\lambda_l\}_{l=1}^{k} \cup \{ x_j \pm i y_j \}_{j=1}^{(N-k)/2}\Big) }
\prod_{j=1}^k w_r^\nu(\lambda_j)\prod_{j=k+1}^{(N+k)/2} w_c^\nu(x_j,y_j)
\end{equation}
where
\begin{equation}
w_r^\nu(\lambda) =\MeijerG[\bigg]{m}{0}{0}{m}{-}{\frac{\nu_1}2, \dots, \frac{\nu_{m-1}}2,0}{\frac{\lambda^2}{2^m}}=
\prod_{j=1}^m \bigg[\int_\R d \lambda^{(j)}\Big(\frac{\lambda^{(j)}}{2}\Big)^{\nu_j/2}e^{-\frac12(\lambda^{(j)})^2}\bigg] \,
\delta( \lambda - \lambda^{(1)} \cdots \lambda^{(m)}).
\end{equation}
\begin{equation}
 w_c^\nu(x,y) = 2\pi\, \int_\R d\delta\,\frac{|\delta|}{\sqrt{\delta^2+4y^2}}\,
 W^\nu\Big(\begin{bmatrix}\mu_+&0\\0&\mu_-\end{bmatrix}\Big)
\end{equation}
with $\mu_\pm$ as in~\eqref{L3a} and
\begin{equation}
 W^\nu(G)  =
 \prod_{l=1}^m\bigg[\int_{\R^{2\times 2}}(dG^{(l)})\det\Big(\frac{G^{(l)}G^{(l)T}}{2}\Big)^{\nu_l/2}
 \frac{e^{-\frac12 \tr G^{(l)}G^{(l)T}}}{\sqrt{2\pi^3}}\bigg]
 \delta(G - G^{(1)} \cdots G^{(m)}).
\end{equation}
The normalisation is given by
\begin{equation}\label{normalisation-rect}
Z_{N}^{m,\nu}=2^{mN(N+1)/4}\prod_{l=1}^m\prod_{j=1}^N\Gamma\Big(\frac{j+\nu_l}{2}\Big).
\end{equation}
\end{prop}

\begin{proof}
The proof follows the same lines as the proof of Theorem~\ref{T3}. We use generalised real Schur decomposition to get an expression for the joint PDF in terms of real eigenvalues and $2\times 2$ matrices, see~\cite[Prop.~4.26]{Ip15thesis}. Finally, changing variables in this expression from the general $2\times 2$ matrix $G$ to a matrix~\eqref{edel-decomp} using an orthogonal similarity transformation and introducing the singular values, $\mu_\pm$, completes the proof.
\end{proof}

\begin{prop}\label{prop:poly}
For the skew-product~\eqref{Km} defined in accordance with the joint PDF given by Proposition~\ref{prop:jpdf-rect}, the polynomials
\begin{equation}\label{skew-ortho-rect}
p_{2j}(x)=x^{2j},\qquad p_{2j+1}=x^{2j+1}-x^{2j-1}\prod_{k=1}^m(2j+\nu_k)
\end{equation}
form a skew-orthogonal set with normalisation
\begin{equation}\label{u-rect}
h_{j-1}^\nu=\prod_{k=1}^m\frac{2\sqrt{2\pi}}{2^{\nu_k}}\Gamma(2j+\nu_k-1).
\end{equation}
\end{prop}

\begin{proof}
For a product square matrices, we found the skew-orthogonal polynomials by exploiting that elements taken of different rows and columns are uncorrelated. This property is still true for rectangular matrices, thus skew-orthogonal polynomials~\eqref{skew-ortho-rect} are obtained following the exact same steps. Likewise for the normalisation~\eqref{u-rect} where we evaluate the generalised partition function~\eqref{L3c1} at $u=v=1$ and use~\eqref{normalisation-rect}.
\end{proof}

With these two propositions established, it is straightforward to extend the rest of our results from square to rectangular matrices. In particularly, we have that the probability of finding exactly $2k$ eigenvalues are real is given by
\begin{equation}
p_{N,2k}^{P_m} = \prod_{l=1}^m\prod_{j=1}^N \frac{1 }{ \Gamma ((j+\nu_l)/2) } 
[\zeta^k]\det \Big [ b_{j,l}^\nu(\zeta)  \Big ]_{j,l=1,\dots,N/2},
\end{equation}
for $N$ even, while the probability of finding $2k+1$ real eigenvalues is
\begin{equation}
p_{N,2k+1}^{P_m} = \prod_{l=1}^m\prod_{j=1}^N \frac{1 }{ \Gamma ((j+\nu_l)/2) }   
[\zeta^k]\det \Big [ [b_{j,l}^\nu(\zeta) ]^{j=1,\dots,(N+1)/2}_{k=1,\dots,(N-1)/2}  
\quad [\tilde a_j^\nu ]_{j=1,\dots,(N+1)/2} \Big],
\end{equation}
for $N$ odd. Here, we have defined
\begin{equation}
b_{j,l}^\nu(\zeta):=(\zeta - 1) \Big ( a_{j,l}^\nu - 2^{-2} \prod_{i=1}^m(2 (l+\nu_i-1)) a_{j,l-1}^\nu \Big ) +
2^{-(2j-1/2)m} h_{j-1}^\nu\delta_{j,l} 
\end{equation}
with $h_{j}^\nu$ given by~\eqref{u-rect}, $a_{j,l}^\nu$ and $\tilde a_j^\nu$ ($j,l>0$) given by~\eqref{aa-rect} and $a_{j,0}^\nu=0$.
The similarity with Theorem~\ref{T1} is immediate.

Moreover, the local densities at the origin is given by
\begin{equation}\label{c-local-rect}
 \frac{2y\,w_c^\nu(x,y)}{(2 \sqrt{2 \pi})^m}\prod_{l=1}^m\frac{\nu_l!}{2^{\nu_l}}\MeijerG{1}{0}{0}{m}{-}{\nu_m,\ldots,\nu_1}{-\abs{z}^2}
\end{equation}
for the complex eigenvalues and 
\begin{equation}\label{r-local-rect}
\int_{-\infty}^\infty dv\,\abs{x-v} w_r^\nu(x)w_r^\nu(v)
\prod_{l=1}^m\frac{\nu_l!}{2^{\nu_l}}\MeijerG{1}{0}{0}{m}{-}{\nu_m,\ldots,\nu_1}{-xv}
\end{equation}
for the real eigenvalues. This generalises~\eqref{pse} and~\eqref{r-local}, respectively. The generalised formulae~\eqref{c-local-rect} and~\eqref{r-local-rect} follows from the derivations in Section~\ref{sec:c-eigenvalues} and~\ref{sec:r-eigenvalues} now using the weights and polynomials from Proposition~\ref{prop:jpdf-rect} and~\ref{prop:poly}. The global densities remains unaltered as long as $\{\nu\}$ are kept fixed in the large-$N$ limit.
 
 \section*{Acknowledgements}
 We would like to thank Mario Kieburg and Oleg Zaboronski comments on this manuscript. Remark~\ref{mario-remark} on page~\pageref{mario-remark} was given to us by Mario Kieburg. 
 The work of PJF was supported by the Australian Research Council grant DP140102613,
 and that of JRI by the ARC Centre of Excellence for Mathematical
 and Statistical Frontiers.
  

\begin{thebibliography}{10}

\bibitem{ARRS13}
K.~Adhikari, N.K. Reddy, T.R. Reddy, and K.~Saha, \emph{Determinantal point
  processes in the plane from products of random matrices}, Ann. Henri Poincare Probab. Stat. \textbf{52} (2016) 16.


\bibitem{Ak05}
G.~Akemann, \emph{The complex {L}aguerre symplectic ensemble of non-{H}ermitian
  matrices}, Nucl. Phys. B \textbf{73} (2005), 253 [arXiv:hep-th/0507156].
  
\bibitem{AB12}
G.~Akemann and Z.~Burda, \emph{Universal microscopic correlations for products
  of independent {G}inibre matrices}, J. Phys. A \textbf{45} (2012), 465210 [arXiv:1208.0187].

\bibitem{AI15}  G.~Akemann and J.R. Ipsen, \emph{Recent exact and asymptotic results for products
of independent random matrices}, Acta Physica Polonica B {\bf 46}  (2015), 1747--1784 [arXiv:1502.01667].
  

\bibitem{AK07}
G.~Akemann and E.~Kanzieper, \emph{Integrable structure of {G}inibre's ensemble
  of real random matrices and a {P}faffian integration theorem}, J. Stat. Phys.
  \textbf{129} (2007), 1159--1231 [arXiv:math-ph/0703019]. 

\bibitem{AKP10}
G.~Akemann, M.~Kieburg, and M.J. Phillips, \emph{Skew-orthogonal {L}aguerre
  polynomials for chiral real asymmetric random matrices}, J. Phys. A
  \textbf{43} (2010), 375207 [arXiv:1005.2983].

\bibitem{APS10}
G.~Akemann, M.J. Phillips, and H.-J. Sommers, \emph{The chiral {G}aussian
  two-matrix ensemble of real asymmetric matrices}, J. Phys. A \textbf{43}
  (2010), 085211 [arXiv:0911.1276].

\bibitem{BEDPSW13}
C.W.J. Beenakker, J.M. Edge, J.P. Dahlhaus, D.I. Pikulin, Shuo Mi, and
  M.~Wimmer, \emph{Wigner-{P}oisson statistics of topological transitions in a
  {J}osephson junction}, Phys. Rev. Lett. \textbf{111} (2013), 037001 [arXiv:1305.2924].

\bibitem{BF12}
G.~Bergqvist and P.J. Forrester, \emph{Rank probabilities for real random $n
  \times n \times 2$ tensors}, Elec. Commun. Probability \textbf{16} (2011),
  630 [arXiv:1106.5581]. 

\bibitem{BS09}
A.~Borodin and C.D. Sinclair, \emph{The {G}inibre ensemble of real random
  matrices and its scaling limit}, Commun. Math. Phys. \textbf{291} (2009),
  177 [arXiv:0805.2986].

\bibitem{BJW10}
Z.~Burda, R.A. Janik, and B.~Waclaw, \emph{Spectrum of the product of
  independent random {G}aussian matrices}, Phys. Rev. E \textbf{81} (2010),
  041132 [arXiv:0912.3422].

\bibitem{Dy62c}
F.J. Dyson, \emph{The three fold way. {Algebraic} structure of symmetry groups
  and ensembles in quantum mechanics}, J. Math. Phys. \textbf{3} (1962),
  1199.

\bibitem{Ed97}
A.~Edelman, \emph{The probability that a random real {G}aussian matrix has $k$
  real eigenvalues, related distributions, and the circular law}, J.
  Multivariate. Anal. \textbf{60} (1997), 203.

\bibitem{EKS94}
A. Edelman, E. Kostlan, and M. Shub, \emph{How many eigenvalues of a random matrix are real?}
J. Amer. Math. Soc. \textbf{7} (1994), 247.
  
\bibitem{Fi72}  
J. L.  Fields, \emph{The asymptotic expansion of the Meijer G-function} Math. Comp. (1972) 757.
  
\bibitem{Fo10}
P.J. Forrester, \emph{Log-gases and random matrices}, Princeton University
  Press, Princeton, NJ, 2010.

\bibitem{Fo13a}
\bysame, \emph{Skew orthogonal polynomials for the real and quaternion
  real {G}inibre ensembles and generalizations}, J.~Phys. A \textbf{46} (2013),
  245203 [arXiv:1302.2638].

\bibitem{Fo14}
\bysame, \emph{Eigenvalue statistics for product complex {W}ishart matrices},
  J.~Phys. A \textbf{47} (2014), 345202 [arXiv:1401.2572].

\bibitem{Fo13}
\bysame, \emph{Probability of all eigenvalues real for products of standard
  Gaussian matrices}, J.~Phys. A \textbf{47} (2014), 065202 [arXiv:1309.7736].
  
\bibitem{Fo15}
\bysame, \emph{Diffusion processes and the asymptotic bulk gap probability for the real Ginibre ensemble}
J. Phys. A \textbf{48} (2015) 324001.

\bibitem{FM08}
P.J. Forrester and A.~Mays, \emph{A method to calculate correlation functions
  for $\beta = 1$ random matrices of odd size}, J. Stat. Phys. \textbf{134}
  (2009), 443 [arXiv:0809.5116].

\bibitem{FM11}
P.J. Forrester and A.~Mays, \emph{Pfaffian point processes for the {G}aussian
  real generalised eigenvalue problem}, Prob. Theory and Rel. Fields
  \textbf{154} (2012), 1 [arXiv:0910.2531].

\bibitem{FN07}
P.J. Forrester and T.~Nagao, \emph{Eigenvalue statistics of the real {G}inibre
  ensemble}, Phys. Rev. Lett. \textbf{99} (2007), 050603 [arXiv:0706.2020].

\bibitem{FN08p}
\bysame, \emph{Skew orthogonal polynomials and the partly symmetric real
  {G}inibre ensemble}, J. Phys. A \textbf{41} (2008), 375003 [arXiv:0806.0055].

\bibitem{FR09}
P.J. Forrester and E.M. Rains, \emph{Matrix averages relating to the {G}inibre
  ensemble}, J. Phys. A \textbf{42} (2009), 385205 [arXiv:0907.0287].

\bibitem{Gi65}
J.~Ginibre, \emph{Statistical ensembles of complex, quaternion, and real
  matrices}, J. Math. Phys. \textbf{6} (1965), 440.

\bibitem{GT10}
F.~G\"otze and A.~Tikhomiroz, \emph{On the asymptotic spectrum of products of
  independent random matrices}, arXiv:1012.2710, 2010.

\bibitem{HJL15}
S.~Hameed, K.~Jain, and A.~Lakshminarayan, \emph{Real eigenvalues of
  non-{G}aussian matrices and their products}, J. Phys. A \textbf{48} (2015),
  385204 [arXiv:1504.06256].

\bibitem{Ip15}
J.R. Ipsen, \emph{Lyapunov exponents for products of rectangular real, complex and quaternionic Ginibre matrices},
J. Phys. A 48 (2015) 155204 [1412.3003].  
  
\bibitem{Ip15thesis}
\bysame, \emph{Products of independent Gaussian random matrices},
  PhD thesis, Bielefeld University (2015) [arXiv:1510.06128].
  
\bibitem{IK14}
J.R. Ipsen and M.~Kieburg, \emph{Weak commutation relations and eigenvalue
  statistics for products of rectangular random matrices}, Phys. Rev. E
  \textbf{89} (2014), 032106 [arXiv:1310.4154].
  
\bibitem{Ka43}
M. Kac, \emph{On the average number of real roots of a random algebraic equation}, 
 Bull. Amer. Math. Soc. 49 (1943), 314.
  
\bibitem{AK05}
E.~Kanzieper and G.~Akemann, \emph{Statistics of Real Eigenvalues in {G}inibre's Ensemble of 
 Random Real Matrices}, Phys. Rev. Lett. 95 (2005) 230201 [arXiv:math-ph/0507058].

\bibitem{KPTTZ15}
E. Kanzieper, M. Poplavskyi, C. Timm, R. Tribe, and O. Zaboronski, \emph{What is the probability that a large random matrix has no real eigenvalues?} arXiv:1503.07926.
 
\bibitem{KB09}
T.G. Kolda and B.W. Bader, \emph{Tensor decompositions and applications}, SIAM
  Review \textbf{51} (2009), 455.

\bibitem{KZ14}
A.B.J. Kuijlaars and L.~Zhang, \emph{Singular values of products of {G}inibre
  matrices, multiple orthogonal polynomials and hard edge scaling limits},
  Comm. Math. Phys. \textbf{332} (2014), 759 [arXiv:1308.1003].

\bibitem{Ku15}
S.~Kumar, \emph{Exact evaluations of some {M}eijer {G}-functions and
  probability of all eigenvalues real for products of two {G}aussian matrices},
  J. Phys. A, \textbf{48} (2015) 445206.

\bibitem{La13}
A.~Lakshminarayan, \emph{On the number of real eigenvalues of products of
  random matrices and an application to quantum entanglement}, J. Phys. A
  \textbf{46} (2013), 152003 [arXiv:1301.7601].

\bibitem{LS91}
N.~Lehmann and H.-J. Sommers, \emph{Eigenvalue statistics of random real
  matrices}, Phys. Rev. Lett. \textbf{67} (1991), 941.

\bibitem{MbA01}
T.~Masser and D.~ben Avraham, \emph{Correlation functions for diffusion-limited
  annihilation $A + A \to 0$}, Phys. Rev. E \textbf{64} (2001), 062101.

\bibitem{Ma11}
A.~Mays, \emph{A geometrical triumvirate of real random matrices}, Ph.D.
  thesis, University of Melbourne, 2012 [arXiv:1202.1218].

\bibitem{Mu82}
R.J. Muirhead, \emph{Aspects of multivariate statistical theory}, Wiley, New
  York, 1982.

\bibitem{NNV15}
H. Nguyen, O. Nguyen, and V. Vu. \emph{On the number of real roots of random polynomials}, 
  Comm. Contemp. Math. (2015) 1550052 [arXiv:1402.4628].
  
\bibitem{OS11}
S.~O'Rourke and A.~Soshnikov, \emph{Products of independent non-{H}ermitian
  matrices}, Electron. J. Prob. \textbf{16} (2011), 2219 [arXiv:1012.4497].
  


  \bibitem{SRL11}
M.S.~Ramkarthik, K.V.~Shuddhodan and A.~Lakshminarayan, \emph{Entanglement
  optimizing mixtures of two-qubit states}, J. Phys. A \textbf{44} (2011),
  345301 [arXiv:0910.4504].
  
\bibitem{Re16} T.R.~Reddy, \emph{Probability that  
product of real random matrices have all eigenvalues real tend to 1}, arXiv:1606.07581, 2016.  

\bibitem{Si06}
C.D. Sinclair, \emph{Averages over {G}inibre's ensemble of random real
  matrices}, Int. Math. Res. Not. \textbf{2007} (2007), rnm015 [arXiv:math-ph/0605006].

\bibitem{Som07}
H.-J. Sommers, \emph{Symplectic structure of the real {G}inibre ensemble},
  J.~Phys. A \textbf{40} (2007), F671 [arXiv:0706.1671].

\bibitem{So07}
A.~Soshnikov, \emph{Statistics of extreme spacing in determinantal random point
  processes}, Mosc. Math. J. \textbf{5} (2007), 705 [arXiv:math/0506286].

\bibitem{tB91}
J.B. ten Berge, \emph{Kruskal's polynomial for $2 \times 2 \times 2$ arrays and
  a generalization to $2 \times n \times n$ arrays}, Psychometrika \textbf{56}
  (1991), 631.

\bibitem{TZ11}
R.~Tribe and O.~Zaboronski, \emph{Pfaffian formulae for one dimensional
  coalescing and annihilating systems}, Elec. J. Prob. \textbf{16} (2011),
  2080 [arXiv:1009.4565].

\end{thebibliography}

\providecommand{\bysame}{\leavevmode\hbox to3em{\hrulefill}\thinspace}
\providecommand{\MR}{\relax\ifhmode\unskip\space\fi MR }
\providecommand{\MRhref}[2]{%
  \href{http://www.ams.org/mathscinet-getitem?mr=#1}{#2}
}
\providecommand{\href}[2]{#2}

\end{document}